\documentclass[aip,jcp,10pt,twocolumn,superscriptaddress]{revtex4-1}  

\usepackage[greek,english]{babel}
\usepackage[utf8x]{inputenc}
\usepackage{amsmath}
\usepackage{amssymb}
\usepackage{graphicx,float,wrapfig}
\usepackage{cancel}
\usepackage{epstopdf}
\usepackage{subfigure} 
\usepackage{color} 
\usepackage{amsthm}

\newtheorem{proposition}{Proposition}[section]
\newtheorem{assumption}{Assumption}[section]

\newcommand{\lat}{\latintext}

\begin{document}

\title{\lat \Large{\color{blue}Parametric Sensitivity Analysis for Stochastic Molecular Systems using Information Theoretic Metrics}}


\author{Anastasios Tsourtis}
\email[E-mail: ]{tsourtis@uoc.gr}
\affiliation{Department of Mathematics and Applied Mathematics, University of Crete, Greece}
\author{Yannis Pantazis}
\email[E-mail: ]{yannis.pantazis@gmail.com}
\author{Markos 	A. Katsoulakis}
\email[E-mail: ]{markos@math.umass.edu}
\affiliation{Department of Mathematics and Statistics, University of Massachusetts, Amherst, USA}
\author{Vagelis Harmandaris}
\email[E-mail: ]{harman@uoc.gr}
\affiliation{Department of Mathematics and Applied Mathematics, University of Crete, Institute of Applied and Computational Mathematics (IACM), Foundation for Research and Technology Hellas (FORTH), GR-70013, Heraklion, Crete, Greece}

\date{\today}

\begin{abstract}
In this paper we extend the parametric sensitivity analysis (SA) methodology proposed in Ref. [Y. Pantazis and M. A. Katsoulakis, J. Chem. Phys. 138, 054115 (2013)]
  
to
continuous time and continuous space Markov processes represented by stochastic differential equations and, particularly,
stochastic molecular dynamics as described by the Langevin equation. The utilized SA method is based on the computation
of the information-theoretic (and thermodynamic) quantity of relative entropy rate (RER) and the associated Fisher information
matrix (FIM) between path distributions. A major advantage of the pathwise SA method is that both RER and pathwise FIM
depend only on averages of the force field therefore they are tractable and computable as ergodic averages from a single run
of the molecular dynamics simulation both in equilibrium and in non-equilibrium steady state regimes. We validate the
performance of the extended SA method to two different molecular stochastic systems, a standard Lennard-Jones fluid and
an all-atom methane liquid and compare the obtained parameter sensitivities with parameter sensitivities on three popular
and well-studied observable functions, namely, the radial distribution function, the mean squared displacement and the pressure.
Results show that the RER-based sensitivities are highly correlated with the observable-based sensitivities.

\end{abstract}

\keywords{Relative Entropy, Sensitivity Analysis, Fisher Information Matrix, Langevin dynamics, Methane Molecular Dynamics}
\maketitle

\section{Introduction}
Molecular simulation is the bridge between theoretically developed models and experimental approaches for the study of molecular systems in the atomistic level. \cite{CompSim,*Frenkel2001} Nowadays, molecular simulation methodologies are used extensively to predict structure-properties relations of complex systems. The importance of numerical simulations in material sciences, biology and chemistry has been
recently acknowledged by the 2013 Nobel price in Chemistry ``for the development of multiscale 
models in complex chemical systems''. The computational modeling of realistic complex molecular systems at the molecular level requires long molecular simulations for an enormous distribution of length and time scales.
\cite{Theodorou_book,Shaw:10,Harmandaris_Fritz, Harmandaris_Johnston_review_softmatter} The properties of the model systems depend on a large number of parameters, which are usually obtained utilizing optimization techniques matching specific data taken either from more detailed (e.g. ab-initio)
simulations or from experiments. Furthermore, stochastic modeling is especially important for
describing the inherent randomness of molecular dynamics in various scales.

All the above complexities imply the need of rigorous mathematical tools for the analysis of both
deterministic and stochastic molecular systems. Uncertainty quantification (UQ) in computational
chemistry is of paramount importance, especially in multiscale modeling, where properties evaluated
at the atomic-molecular scale are transferred to the mesoscopic scale. \cite{Harmandaris_Fritz,Chernatynskiy_review}
Sources of epistemic uncertainty can stem from (i) numerical uncertainty, (ii) model uncertainty and (iii) parametric uncertainty.
Numerical uncertainties are related to the finite time of the dynamic simulation, the number of particles, as well as values of the parameters related to the numerical method used (e.g. time-step), to name some.
Model uncertainty comes from the specific force field representation and its calibration to experimental properties, and the usage of specific boundary conditions.
Parametric uncertainties stem from errors in parameter values due to noisy or insufficient measurements.
Of all the above, the uncertainty associated with the parameters of the potential is the least understood. \cite{Chernatynskiy_review, Koumoutsakos}
Furthermore, intrinsic stochasticity of the system is added
on top of epistemic uncertainty. This type of uncertainty is also called aleatoric.


There exists a diverse range of UQ approaches proposed in the literature. Variance-based methods
such as analysis of variance (ANOVA), \cite{Liu_Chen} Bayesian statistical analysis \cite{Pernot,Koumoutsakos} (applications to deal with large uncertainties) and polynomial chaos expansions \cite{Rizzi} have been widely used. The first two methods
are based on multiple and usually expensive Monte Carlo runs resulting in huge computational cost whereas the latter becomes intractable when the parameter space is large. An in depth study of this last method in MD has recently been presented by Rizzi et al. \onlinecite{Rizzi}.

Sensitivity analysis (SA) is a powerful tool that gives insight of how small variations (uncertainty) in
system parameters (input), can affect the output of the system substantially. Such perturbations occur
from computational errors, uncertainty and errors resulting from experimental parameter estimation
\cite{Rao} (such as parameter fitting through ensemble averages of macroscopic thermodynamic
quantities). Thus, parametric SA can provide critical insights in uncertainty quantification. Especially
in the stochastic setting (e.g., Langevin dynamics in molecular systems), SA is performed by analysis of the system's mean
behavior, i.e., several simulations starting from a configuration at the stationary (or steady states) regime.
The stationary regime is crucial for complex molecular systems since it captures not only static quantities
such as the radial distribution function but also dynamical quantities which includes transitions between metastable
states in complex, high-dimensional energy landscapes and intermittency. \cite{PantazisRER} Depending
on the magnitude of the perturbations, SA can be classified into local (infinitesimal, one-at-a-time
parameter perturbation) and global (finite, multiple parameter perturbation).

Furthermore, the role of SA is not restricted to UQ but it is of pivotal significance in several other applications.
First, robustness of a system meaning the stability of the behavior under simultaneous changes in
model parameters 
or variations of orders of magnitude in insensitive parameters that insignificantly affect the dynamics
 can be addressed utilizing parametric SA approaches.
Second, sensitivity analysis on experiment
conditions under which information loss is minimized, establish optimal experimental design.
\cite{Emery:98} Furthermore, identifiability analysis employs SA to determine a priori whether certain
parameters can be estimated from experimental data of a given type. The work in Ref. ~\onlinecite{Cooke}
contains a general framework of SA in MD (proteins) using the observable helicity
while cross-validation with experimental data is also displayed.
Overall, SA plays a fundamental role in multiscale design and as it has been highlighted by Braatz et al.\cite{Braatz:06}.

Typically in a stochastic setting, the most common local parametric SA method is based on partial
derivatives on ensemble averages of quantities of interest around a nominal parameter value. \cite{Gunawan}
Large derivatives indicate strong sensitivity of the observable to the particular parameter while
the opposite holds when the derivative values are small. There has been an increasing number
of methods to compute the partial derivatives especially in discrete-event systems whose applications
range from biochemical reaction networks to operations
research and queuing theory. Finite-difference approaches based on common random numbers,
\cite{Gunawan} on common reaction path \cite{Sheppard_Khammash} which exploits positive
correlations among coupled perturbed/unperturbed reaction paths as well as coupling methods \cite{Anderson:12, AK:2013}
have been proposed. There are certain issues associated with these finite-difference approaches;
the estimator of the partial derivative has bias while the variance of the gradient estimator increases
with the dimension of the parameter space. Instead of using the finite-difference approaches, one can
utilize Girsanov measure transformation to directly compute the infinitesimal sensitivity.
\cite{Glynn:90, Plyasunov, Warren_Allen}
In MD simulations where both time and states are typically continuous, Iordanov et al. \cite{Iordanov}
performed SA in three potentials of different functional forms (and number
of parameters) in order to compare their influence in thermodynamic quantities. Instead of perturbing
the potential parameters
 they scaled the potentials one-at-a-time (local SA) aiming to minimize the discrepancy from the experimental values
of each observable separately.

Another class of sensitivity methods which is not focused on specific observable functions but on the
overall properties of the stochastic process is based on information theory concepts. Application of information-theoretic
SA methods to analysis of stochastic models uses quantities such as entropy, relative entropy (or Kullback-Leibler
divergence), the corresponding Fisher information matrix as well as mutual information. Relative entropy
(RE) measures the inefficiency of assuming a perturbed (or "wrong") distribution instead of assuming the
unperturbed (or "true") one. RE have been used for the SA study of climate models \cite{Majda_Gershgorin_2010}
where the equilibrium probability density function (PDF) has been obtained through an entropy maximization subject
to constraints induced by the measurements, while Fisher information matrix (FIM) is an indispensable tool for
optimal experimental design. \cite{Emery:98}
In a  typical SA approach based on RE,  explicit knowledge of the equilibrium PDF is assumed. However, in systems with non-equilibrium steady states
(NESS) (i.e., systems in which a steady state is reached but the detailed balance condition is violated), there are
no explicit formulas for the stationary distribution and even when a Gibbs measure is available, it is usually
computationally inefficient to sample from. Such non-equilibrium systems are common place in molecular systems with multiple mechanisms such as reaction-diffusion systems or driven molecular systems. \cite{Harman_non_eq}

In Ref.~\onlinecite{PantazisRER}, the RE between path distributions (i.e., distributions of the particle trajectories) for
discrete time or discrete event systems was been utilized as a measure of sensitivity. When the system is in
stationarity the relative entropy of the two path distributions decomposes into two parts; (i) the relative entropy
rate (RER) that scales linearly with time and (ii) a constant term related to the relative entropy of the initial
distribution of the system. In long times, the first term dominates providing major insights on the sensitivity
of the system with respect to parameter perturbations. In this work, we extend the SA method proposed in
Ref.~\onlinecite{PantazisRER} to the case of stochastic differential equations (i.e., continuous-time and continuous
state space) and particularly the Langevin equation. Furthermore, when perturbations are small, a Taylor
expansion on RER is performed revealing the lower order of this expansion which is the pathwise FIM associated
with the RER. Practically, RER is an observable of the stochastic process which can be computed numerically as
an ergodic average in a straightforward manner as it only requires local dynamics (in our case the forces). Similarly,
FIM computations are feasible in the same fashion with the advantage of being more informative since any
perturbation direction can be explored. Both of these observables can be sampled on the fly, from a single MD run
since only the reference process needs to be simulated. Finally, spectral methods for the calculation of RER were introduced for the over-damped Langevin case in Ref.~\onlinecite{Haas:2013}.


The studied pathwise SA method has major advantages which can be listed as follows. First, it is gradient-free
method which does not require knowledge of the equilibrium PDF. By gradient-free we mean that the pathwise FIM does
not depend on the extent of the perturbation, so when the computation for different perturbations is necessary
(especially in high-dimensional problems) the extra cost is minimal in comparison to the straightforward RER calculation.
Second, it is rigorously valid for long-time, stationary dynamics in path-space including metastable dynamics in a
complex landscape. Third, it is suitable for non-equilibrium systems from statistical mechanics perspective; for
example in NESS processes such as dissipative systems where the structure of the equilibrium PDF is unknown.
Fourth, it is fast since it requires samples only from the unperturbed process which can be also obtained
in a trivially  parallel manner.



The work presented here is a part of a more general hierarchical simulation scheme that involves multiple
simulation level and a broad range of length and time scales. \cite{Johnston13,Rissanou14,Harmandaris14} In this work, we apply the above
methodology on stochastic molecular systems as the RER and FIM methods are based on this setting i.e., non-deterministic with random noise.
As test cases we examine: (a) a benchmark Lennard-Jones (LJ) fluid model, \cite{Frenkel2001,Smit}
and, (b) a detailed all-atom methane ($CH_4$) model. \cite{Dreiding}
The LJ fluid system is the most widely used in molecular simulations model of a simple fluid, whereas
the second one employs more complexity due to the intramolecular bond and angle
potentials in addition to the LJ intermolecular potential. Methane has been also extensively studied over the
years due to the fact that it is in abundance in nature and has environmental impacts as well
as it can be used as fuel being the main component of natural gas.

The proposed pathwise SA method is validated through proper observable quantities upon perturbation of the
potential parameters, which include structural, dynamical and thermodynamic properties of both LJ and methane
model systems. We stress here that the utilized SA method based on RER and the pathwise FIM is independent
of the observable quantities, which is not the case
for derivative-based SA methods where they suffer from smoothness assumptions on the observable functionals.
The partial derivative of an observable is related with the RE through the Pinsker inequality (see ineq. (\ref{Pinsker})).
The Pinsker inequality asserts that small RER (or FIM) values result in small changes  in observable
expectation values under perturbation; thus RER and FIM can serve as a screening tool for specific observables.
The present work provides a detailed quantitative study concerning the relation between the pathwise SA method
(RER / FIM tools) and specific observables of molecular systems. Specifically, it can be inferred that the parameter
that controls the well's position of the LJ potential (i.e., $\sigma_{LJ}$) is two orders of magnitude more sensitive
than the strength of the LJ potential (i.e., $\epsilon_{LJ}$) in terms of RER for the LJ fluid model as well as for the
methane model. This result is also confirmed by the observables. Moreover, we show that $5\%$ perturbation of
$\sigma_{LJ}$ is more crucial in terms of RER (but also supported from the observable functions) than changing
the potential cutoff radius from $4\sigma_{LJ}$ to $1.6\sigma_{LJ}$. Interestingly, the use of RER as an information
criterion to assign values to simulation parameters such as the cutoff radius may result in accelerated MD simulations
with the induced error being controlled by the RER. Finally, the intramolecular parameters (i.e., the parameters of the
bonds) in the methane model are more sensitive in terms of RER than the intermolecular parameters which is
explained by the fact that the bond and angle forces are stiffer than the between-atoms LJ forces.

The organization of the paper is as follows. The following Section describes the path-wise sensitivity analysis
method for Langevin dynamics in detail. In Section~\ref{Models}, the LJ fluid model, the methane model as
well as various observable functions are presented followed by Section~\ref{results_section} where the validation
of the proposed pathwise SA method is demonstrated. Finally, we conclude the paper in Section~\ref{conclusions_section}.

\section{Pathwise Sensitivity Analysis for Langevin Dynamics}
\label{pathwise_SA_meth}
This Section describes and motivates the info-theoretic approach for sensitivity analysis of
stochastic Molecular Dynamics. Particularly, the RER and the corresponding pathwise FIM
are derived for the Langevin equation.

\subsection{Stochastic equation of motion}
Langevin dynamics models a Hamiltonian system which is coupled with a thermostat. \cite{FreeEnergy}
The thermostat serves as a reservoir of energy. In Langevin dynamics, the motion of particles is governed
through a probabilistic framework by a system of stochastic differential equations given by
\begin{equation}
 \left\{
 \begin{array}{l}
	d{q}_t = M^{-1}{p}_t dt \\
	d{p}_t = F^\theta({q}_t) dt - \gamma M^{-1}{p}_t dt + \sigma d{W}_t \ ,
\end{array}  \right.
\label{eq:system}
\end{equation}
where ${q}_t\in \mathbb R^{dN}$ is the position vector of the $N$ particles in $d$-dimensions,
${p}_t\in \mathbb R^{dN}$ is the momentum vector of the particles, $M$ is the (diagonal) mass matrix,
$F^\theta(\cdot):\mathbb R^{dN}\rightarrow\mathbb R^{dN}$ is the driving (conservative) force which depends on
a parameter vector $\theta\in\mathbb R^K$ (e.g. parameters of the specific atomistic force field), $\gamma$ is the friction matrix,
$\sigma$ is the diffusion matrix and ${W}_t$ is a $dN$-dimensional Brownian motion.
In the equilibrium regime, the forces are of gradient form, i.e., $F^\theta({q}_t) = - \nabla V^\theta({q_t})$
where $V^\theta(\cdot)$ is the potential energy.  Moreover, the fluctuation-dissipation
theorem asserts that friction and diffusion terms are related with the inverse temperature $\beta\in\mathbb R$
of the system by
\begin{equation*}
\sigma\sigma^T = 2\beta^{-1}\gamma\, .
\end{equation*}
Under gradient-type forces and the fluctuation dissipation theorem, the Langevin system has a Gibbs
equilibrium (or invariant) distribution, $\mu^\theta(\cdot,\cdot)$, given by
\begin{equation}
\mu^\theta(dq,dp) = \frac{1}{Z}e^{-\beta(V^\theta(q)+\frac{1}{2}p^TM^{-1}p)} dqdp \ .
\label{eq:stat_meas}
\end{equation}
In non-equilibrium steady states, however, the stationary distribution, $\mu^\theta(\cdot,\cdot)$, is generally not known
restricting the sensitivity analysis methods that rely on the explicit knowledge of the steady states.
Though, as we show below, the proposed pathwise sensitivity methodology is not limited to equilibrium
systems and it works equally well in the non-equilibrium steady states regime since it only necessitates
the explicit knowledge of the driving forces (i.e., the local dynamics).

\subsection{Relative Entropy Rate and Fisher Information Matrix for Langevin Processes}
Let the path space $\mathcal X$ be the set of all trajectories $\{({q}_t, {p}_t) \}_{t=0}^T$
generated by the Langevin equation in the time interval $[0,T]$. Let $Q_{[0,T]}^\theta$ denote the path space distribution,
i.e., the probability to see a particular element of path space, $\mathcal X$, for a specific set of parameters $\theta$.
Consider also a perturbation vector,
$\epsilon_{0} \in \mathbb{R}^{K}$, and denote by $Q_{[0,T]}^{\theta+\epsilon_{0}}$ the path space distribution
of the perturbed process, $(\bar{q}_t, \bar{p}_t)$. The proposed sensitivity analysis approach is based on the quantification of the difference
between the two path space probability distributions by computing the relative entropy (RE) between them. Thus, the
pathwise RE of the unperturbed distribution, $Q_{[0,T]}^{\theta}$, with respect to (w.r.t.) the perturbed distribution,
$Q_{[0,T]}^{\theta+\epsilon_{0}}$, assuming that they are absolutely continuous w.r.t. each other is defined as
\begin{equation}
	\mathcal{R}(Q_{[0,T]}^{\theta}| Q_{[0,T]}^{\theta+\epsilon_{0}})
	:= \int \log{\Bigg( \frac{dQ_{[0,T]}^{\theta}}{dQ_{[0,T]}^{\theta+\epsilon_{0}}} \Bigg)} dQ_{[0,T]}^{\theta} \ ,
	\label{pathwise:RE}
\end{equation}
where $\frac{dQ_{[0,T]}^{\theta}}{dQ_{[0,T]}^{\theta+\epsilon_{0}}}$ is the Radon-Nikodym derivative
and it is well-defined due to the absolute continuity assumption. A key property of RE  is that
$\mathcal{R}(Q_{[0,T]}^{\theta}| Q_{[0,T]}^{\theta+\epsilon_{0}}) \ge 0$ with equality if and only if
$Q_{[0,T]}^{\theta}=Q_{[0,T]}^{\theta+\epsilon_{0}}$,  which  allows us to view relative entropy  as a
``distance" (more precisely a semi-metric) between two probability measures capturing the relative
importance of parameter vector changes. \cite{Liu_Chen} Moreover, from an information theory
perspective, the relative entropy measures {\it loss/change of information}
when $Q_{[0,T]}^{\theta+\epsilon_{0}}$ is considered instead of $Q_{[0,T]}^{\theta}$. \cite{CoverThomas}

The necessary and sufficient conditions of the two path distributions (perturbed and unperturbed)
to be absolutely continuous are provided next.
\begin{assumption} Assume that
\begin{itemize}
\item[(a)] the diffusion matrix, $\sigma$, is invertible, and,
\item[(b)] $\mathbb E_{Q_{[0,T]}^{\theta}} [\exp\big\{\int_0^T |u(q_t,p_t)|^2  \big\}]<\infty$, where
the function $u(\cdot,\cdot):\mathbb R^{2dN}\rightarrow\mathbb R^{2dN}$ is defined
such that for all pairs $(q,p)$ it should hold that {\small
\begin{equation*}
	\left[ \begin{array}{cc} 0 & 0 \\ 0 & \sigma \end{array} \right] u(q,p)
	= \left[ \begin{array}{l}
	 M^{-1}p - M^{-1}p \\
	 F^\theta(q) - \gamma M^{-1} p - (F^{\theta+\epsilon_{0}}(q) - \gamma M^{-1} p)
	\end{array} \right ] \ ,
\end{equation*}
}or, equivalently, {\small
\begin{equation*}
 	\left[ \begin{array}{c} 0  \\ \sigma \end{array} \right] u(q,p)
	= \left[ \begin{array}{l}
	 0 \\	F^\theta(q) - F^{\theta+\epsilon_{0}}(q)
	\end{array} \right]  \ .
\end{equation*}}
\end{itemize}
\label{abs:cont:ass}
\end{assumption}
Notice that such a function, $u(\cdot,\cdot)$, exists due to (a). Furthermore, (a) implies that
the noise is non-degenerate for the momenta.
Then, the RE of the path distribution defined in (\ref{pathwise:RE}) is finite and an explicit formula
can be estimated as the following proposition asserts.

\begin{proposition}\label{propos:RE}
Let Assumption~\ref{abs:cont:ass} holds. Assume also that $({q}_0, {p}_0)\sim\nu^{\theta}$
and $(\bar{q}_0, \bar{p}_0)\sim\nu^{\theta+\epsilon_{0}}$ where $\nu^{\theta}(\cdot,\cdot)$
and $\nu^{\theta+\epsilon_{0}}(\cdot,\cdot)$
are two initial distributions which should be absolutely continuous w.r.t. each other. Then,
\begin{equation}
\begin{aligned}
\mathcal{R}(Q_{[0,T]}^{\theta}| Q_{[0,T]}^{\theta+\epsilon_{0}}) &= \mathcal{R}(\nu^\theta | \nu^{\theta+\epsilon_{0}}) \\
&+ \frac{1}{2} \mathbb E_{Q_{[0,T]}^{\theta}} \Big[ \int_0^T |u(q_t,p_t)|^2 dt \Big]
\label{pathwiseRE:decomp}
\end{aligned}
\end{equation}
\end{proposition}

\begin{proof}
Under Assumption~\ref{abs:cont:ass}, the Girsanov theorem applies providing an explicit formula of the Radon-Nikodym
derivative\cite{Oksendal} which is given by {\small
\begin{equation*}
\begin{aligned}
&\frac{dQ_{[0,T]}^{\theta}}{dQ_{[0,T]}^{\theta+\epsilon_{0}}}\Big(\big\{(q_t,p_t)\big\}_{t=0}^T \Big)
= \frac{d\nu^\theta}{d\nu^{\theta+\epsilon_{0}}}(q_0,p_0) \times \\
&\exp\left\{-\int_0^T u(q_t,p_t)^TdW_t - \frac{1}{2}\int_0^T |u(q_t,p_t)|^2 dt \right\} \ .
\end{aligned}
\end{equation*}}
Moreover, $\hat{W}_t:=\int_0^t u(q_s,p_s)dt + W_t$ is a Brownian motion with respect to the path distribution
$Q_{[0,T]}^{\theta}$, meaning that, for any measurable function $f(\cdot,\cdot)$, it holds
$\mathbb E_{Q_{[0,T]}^{\theta}} \big[\int_0^T f(q_t,p_t)^Td\hat{W}_t\big] = 0$. Then, {\small
\begin{equation*}
\begin{aligned}
&\mathcal{R}(Q_{[0,T]}^{\theta}| Q_{[0,T]}^{\theta+\epsilon_{0}}) =
\int \left(\log \frac{d\nu^\theta}{d\nu^{\theta+\epsilon_{0}}}(q_0,p_0) \right. \\
&- \left. \int_0^T u(q_t,p_t)^TdW_t - \frac{1}{2}\int_0^T |u(q_t,p_t)|^2 dt \right) dQ_{[0,T]}^{\theta} \\
&= \int \log \frac{d\nu^\theta}{d\nu^{\theta+\epsilon_{0}}}(q_0,p_0) dQ_{[0,T]}^{\theta}
- \int \int_0^T u(q_t,p_t)^Td\hat{W}_t dQ_{[0,T]}^{\theta} \\
&+ \frac{1}{2} \int \int_0^T |u(q_t,p_t)|^2 dt dQ_{[0,T]}^{\theta} \\
&= \mathcal{R}(\nu^\theta | \nu^{\theta+\epsilon_{0}}) + \frac{1}{2} \int \int_0^T |u(q_t,p_t)|^2 dt dQ_{[0,T]}^{\theta}
\end{aligned}
\end{equation*}}
\end{proof}

We remark that this proposition is a result on the transient regime since the initial distributions can be anything
as fas as they are absolutely continuous w.r.t. each other.
In the stationary regime, a significant simplification of the pathwise RE occurs. As the following proposition
asserts, pathwise RE is decomposed into a linear in time term plus a constant where the slope of the linear
term is the relative entropy rate (RER).

\begin{proposition}\label{propos:RER}
Let Assumption~\ref{abs:cont:ass} holds. Assume also that $({q}_0, {p}_0)\sim\mu^{\theta}$
and $(\bar{q}_0, \bar{p}_0)\sim\mu^{\theta+\epsilon_{0}}$ where $\mu^{\theta}(\cdot,\cdot)$
and $\mu^{\theta+\epsilon_{0}}(\cdot,\cdot)$ are the stationary distributions for the unperturbed
and the perturbed process, respectively, which should be absolutely continuous w.r.t. each other.
Then, the pathwise RE equals to
\begin{equation}
\mathcal{R}(Q_{[0,T]}^{\theta}| Q_{[0,T]}^{\theta+\epsilon_{0}}) = T \mathcal{H}(Q^{\theta} | Q^{\theta+\epsilon_0})
+ \mathcal{R}(\mu^\theta | \mu^{\theta+\epsilon_0})
\label{eq:discr_RE_exp}
\end{equation}
where {\small
\begin{equation}
\begin{aligned}
	&\mathcal{H}(Q^{\theta} | Q^{\theta+\epsilon_0}) :=  \\ &\frac{1}{2} \mathbb{E}_{\mu^\theta}
	[ ( F^{\theta+\epsilon_0}({q}) - F^\theta({q}))^T (\sigma\sigma^T)^{-1} ( F^{\theta+\epsilon_0}({q}) - F^\theta({q}))]
	\label{eq:RER_cont}	
\end{aligned}
\end{equation} }
is the Relative Entropy Rate.
\end{proposition}

\begin{proof}
First notice that we drop the $T$ subscript from the definition of RER because RER is time-independent.
Then, it is straightforward to show from the previous proposition that {\small
\begin{equation*}
\begin{aligned}
&\mathcal{R}(Q_{[0,T]}^{\theta}| Q_{[0,T]}^{\theta+\epsilon_{0}}) \\
&= \mathcal{R}(\mu^\theta | \mu^{\theta+\epsilon_{0}}) + \frac{1}{2} \int \int_0^T |u(q_t,p_t)|^2 dt dQ_{[0,T]}^{\theta} \\
&=  \mathcal{R}(\mu^\theta | \mu^{\theta+\epsilon_{0}}) + \frac{1}{2} \int_0^T \int |u(q_t,p_t)|^2 dQ_{[0,T]}^{\theta} dt \\
&=  \mathcal{R}(\mu^\theta | \mu^{\theta+\epsilon_{0}}) + \frac{1}{2} \int_0^T \int |u(q,p)|^2 \mu^{\theta}(dq,dp) dt \\
&= \mathcal{R}(\mu^\theta | \mu^{\theta+\epsilon_{0}}) + \frac{T}{2} \mathbb{E}_{\mu^\theta}[ |u(q,p)|^2 ] \ .
\end{aligned}
\end{equation*}}
\end{proof}
RER inherits all the properties of relative entropy (non-negativity, convexity, etc.) and it
measures the change of information in path space per unit time. For large times, the term that involves RER is the significant
term, since it scales linearly with time, while the constant one becomes less and less important.
Moreover, the estimation of RER necessitates only the knowledge of the driving forces (i.e., the local
dynamics) which is available since the driving forces are computed in any numerical scheme of the
Langevin equation.

\medskip
\noindent
{\bf Pathwise Fisher information matrix}:
Generally, RE is locally a quadratic functional in a neighborhood of parameter vector, $\theta$.
Under smoothness assumption in the parameter vector, the curvature
of the RE around $\theta$, defined by its Hessian, is the FIM. Analogously, we define the Hessian
of the RER to be the pathwise FIM denoted by $F_{\mathcal{H}}(Q^{\theta})$. The relation between
the RER and the pathwise FIM is
\begin{equation}
	\mathcal{H}(Q^{\theta}| Q^{\theta+\epsilon_{0}}) = \frac{1}{2}\epsilon_{0}^{T}F_{\mathcal{H}}(Q^{\theta})\epsilon_{0} + \mathcal{O}(|\epsilon_{0}|^3) \ .
\label{eq:RE_FIM_expansion}	
\end{equation}
Under smoothness assumption of the force vector, $F^{\theta}(\cdot)$, w.r.t. to the parameter vector,
$\theta$, an explicit formula for the pathwise FIM for the Langevin process is straightforwardly obtained from
(\ref{eq:RER_cont}) given by
\begin{equation}
	{F}_{\mathcal H}(Q) = \mathbb{E}_{\mu^\theta}[\nabla_{\theta}F^{\theta}(q)^T(\sigma\sigma^T)^{-1}\nabla_{\theta}F^{\theta}(q)] \ ,
	\label{eq:FIM_cont}
\end{equation}
where $\nabla_{\theta}F^{\theta}(\cdot)$ is a $dN\times K$ matrix containing all the first-order partial
derivatives of the force vector (i.e., the Jacobian matrix). Observe that the pathwise FIM does not
depend on the perturbation vector, $\epsilon_0$, making pathwise FIM an attractive ``gradient-free''
quantity for sensitivity analysis. Indeed, the RER for any perturbation can be recovered up to third-order
utilizing only the pathwise FIM and (\ref{eq:RE_FIM_expansion}). Moreover, the spectral analysis
of ${F}_{\mathcal H}(Q)$ would allow to identify which parameter directions are most/least sensitive
to perturbations.


\medskip
\noindent
{\bf Example 1: Unknown stationary distribution}:
In many molecular systems the steady state is not  a Gibbs distribution and typically it is not known explicitly. 
This is commonplace in  non-equilibrium molecular systems such as models  with multiple mechanisms, e.g.  reaction-diffusion systems,  or driven molecular systems. \cite{Harman_non_eq, VH_non_rev_shear2000}
Here we consider such a  mathematically simple  example, where we assume that the force field consists of two components; one conservative term given as minus the gradient
of the potential energy and another term that is not the gradient of a potential function. Mathematically, the
force field is given by
\begin{equation*}
F^\theta(q) = -\nabla V^\theta(q) + G(q)
\end{equation*}
where we further assume for simplicity that only the conservative term depends on the parameter vector, $\theta$.
Since, the resulting Langevin process is at the non-equilibrium regime, the steady states do not admit an explicit
form. However, denoting by $\bar{\mu}^\theta$ the unknown stationary distribution of the Langevin process driven
by the above forces, the RER is given by {\small
\begin{equation}
\begin{aligned}
	& \mathcal{H}(Q^{\theta} | Q^{\theta+\epsilon_0}) =  \\ &\frac{1}{2} \mathbb{E}_{\bar{\mu}^\theta}
	[ ( \nabla V^{\theta+\epsilon_0}({q}) - \nabla V^\theta({q}))^T (\sigma\sigma^T)^{-1} ( \nabla V^{\theta+\epsilon_0}({q}) - \nabla V^\theta({q}))] \ .
\end{aligned}
\end{equation}}
Notice that the expression in the expectation does not depend on the non-conservative forces and it is
the same expression as in the equilibrium regime. However, the dependence on the non-conservative
forces is evident through the (unknown) stationary distribution, $\bar{\mu}^\theta$.

\medskip
\noindent
{\bf Example 2: Inverse temperature perturbation}:
Using the fluctuation-dissipation relation, we can substitute the friction parameter $\gamma$ with the
inverse temperature $\beta$ and compute the RER and the pathwise FIM for $\beta$ perturbations.
Indeed, substituting in eq. (\ref{eq:RER_cont}) the relation $\gamma=\frac{1}{2}\beta\sigma\sigma^T$, we
are looking for $u(\cdot,\cdot)$ such that {\small
\begin{equation*}
 	\left[ \begin{array}{cc} 0 & 0  \\ 0 & \sigma \end{array} \right] u(q,p)
	= \left[ \begin{array}{l}
	 0 \\ -\frac{1}{2}\beta \sigma\sigma^T M^{-1}p + \frac{1}{2}(\beta+\epsilon_\beta) \sigma\sigma^T M^{-1}p
	\end{array} \right]  \ ,
\end{equation*}}
where $\epsilon_\beta$ is the perturbation of inverse temperature. Notice that the forces were cancelled out
in this expression for $u$ because no perturbation is performed in the parameters of the forces. At the stationary
regime, RER is then given by
\begin{equation}
\mathcal{H}(Q^{\beta} | Q^{\beta+\epsilon_\beta}) = \frac{\epsilon_\beta^2}{8} \mathbb{E}_{\mu^\beta} [p^TM^{-1}\sigma\sigma^TM^{-1}p] \ ,
\label{RER:inv:temp}
\end{equation}
where $\mu^\beta(\cdot)$ is the stationary distribution of the process. It is evident that RER is a quadratic function
of the perturbation of the inverse temperature and interestingly enough it depends only on the momenta, $p$. The
above formula is valid for any force field and implies that the sensitivity of the (inverse) temperature as quantified by the
relative entropy between path distributions is independent of the underlying system as it defined by the forces or by the
potential function, $V^\theta(\cdot)$.

Furthermore, in the equilibrium regime where the stationary distribution
is given by the Gibbs measure (eq. (\ref{eq:stat_meas})), (\ref{RER:inv:temp}) can be further simplified because of the
Gaussian nature of the momenta, $p$. Indeed, assuming
for simplicity that $M=m I_{dN}$ and $\sigma=\sigma I_{dN}$ with $m, \sigma \in \mathbb R$, (\ref{RER:inv:temp})
is rewritten as
\begin{equation}
\mathcal{H}(Q^{\beta} | Q^{\beta+\epsilon_\beta}) = \frac{\epsilon_\beta^2\sigma^2}{8\beta m} dN \ .
\end{equation}
Consequently, the pathwise FIM in the logarithmic scale (see equation below) is given by
\begin{equation}
F_{\mathcal{H}}(Q^{\log{\beta}}) =  \frac{\gamma}{2m} dN \ .
\end{equation}

\medskip
\noindent
{\bf SA in the logarithmic scale}:
In many molecular systems, the model parameters may differ by orders of magnitude, thus, it is more appropriate to perform
relative perturbations, i.e., the $i$-th element of the perturbation vector is $\theta_i\epsilon_{0,i}$. After straightforward
algebra, the elements of the logarithmic-scale Fisher information matrix are given by {\small
\begin{equation}
	( F_{\mathcal{H}}(Q^{\log{\theta}}) )_{i,j} = \theta_{i} \theta_{j} ( F_{\mathcal{H}}(Q^{\theta}) )_{i,j}\ , \quad i,j=1,\dots,K \ .
\end{equation}}
We refer to Ref.~\onlinecite{PantazisRER} for more details.

\medskip
\noindent
{\bf Statistical estimators}:
Even though the Langevin equation is degenerate since the noise applies only to the momenta, the
process is hypo-elliptic and ergodic under mild conditions on the potential energy, $V(\cdot)$.
Therefore, RER and the corresponding pathwise FIM can be computed as ergodic averages.
Note though that in order to obtain samples from the Langevin process, a numerical scheme
should be employed resulting in errors due to the discretization procedure. There exist several
numerical integrators such as BBK and BAOAB for the Langevin equation. \cite{FreeEnergy,Leimkuhler:13}
In Appendix~\ref{apdx:discrete}, BBK integrator is briefly reviewed. The inserted
bias is of order $O(\Delta t)$ where $\Delta t$ is the time-step as it has been shown for Langevin
equation under compactness condition \cite{Bally:96b, Bally:96a, Mattingly:10} (e.g., under
bounded domain). Then, the statistical estimator for the RER is given by {\small
\begin{equation}
\begin{aligned}
\bar{\mathcal H}(Q^{\theta} | Q^{\theta+\epsilon_0}) =
&\frac{1}{2n} \sum_{i=1}^{n} \big(F^{\theta+\epsilon_0}(q^{(i)})-F^{\theta}(q^{(i)})\big)^T \\
& (\sigma\sigma^T)^{-1}
\big(F^{\theta+\epsilon_0}(q^{(i)})-F^{\theta}(q^{(i)})\big) \ ,
\end{aligned}
\label{eq:H_cont}
\end{equation}}
where $n$ is the number of samples and, similarly, the statistical estimator for the pathwise FIM
is given by {\small
\begin{equation}
\bar{F}_{\mathcal{H}}(Q^{\theta}) = \frac{1}{n} \sum_{i=1}^{n}
\nabla_{\theta}F^{\theta}(q^{(i)})^{T}(\sigma\sigma^T)^{-1}\nabla_{\theta}F^{\theta}(q^{(i)}) \ .
\label{eq:F_cont}
\end{equation}}

\medskip
\noindent
{\bf Sensitivity Bound}:
Relative entropy provides a mathematically elegant and computationally tractable  methodology for the
parameter sensitivity analysis of Langevin systems. Such an approach focuses on the sensitivity of the entire
probability distribution, either at equilibrium or at the path-space level, i.e., for the entire stationary
time-series quantifying among others the transferability skills of the molecular models.
However, in many situations in molecular simulations, the interest is focused on
observables such as radial distribution function, pressure, mean square displacement, etc.
Therefore, it is desirable to attempt to connect the parameter sensitivities of  observables to the relative entropy
methods proposed here. Indeed, relative entropy can provide an upper bound for a large family
of observable functions, $g$, through the Pinsker (or Csiszar-Kullback-Pinsker) inequality, \cite{CoverThomas}{\small
\begin{equation}
	|\mathbb{E}_{Q_{[0,T]}^{\theta + \epsilon_{0}}}[g] - \mathbb{E}_{Q_{[0,T]}^{\theta}}[g]|
	\leq ||g||_{\infty} \sqrt{2\mathcal{R}(Q_{[0,T]}^{\theta}|Q_{[0,T]}^{\theta+\epsilon_{0}})}
	\label{Pinsker}
\end{equation}}
where $||\cdot||_\infty$ denotes  the supremum (here, maximum) of $g$.
In the context of sensitivity analysis, inequality (\ref{Pinsker}) states that if the relative entropy is small, i.e., insensitive in a particular parameter  direction, then, any  bounded observable $g$ is also expected to be  insensitive towards the same direction.
In this sense, ineq. (\ref{Pinsker}) can be viewed as a screening tool for parametric "insensitivity analysis" of observables. Sensitivity bounds sharper than ineq. (\ref{Pinsker}) are also being developed. \cite{DKPP:2014}
Note that the inverse is not justified; i.e. if the relative entropy is large for a specific parameter direction, then an observable $g$ might, or might not, exhibit sensitivity with respect to the same parameter direction.

\section{Models and Observables}
\label{Models}
This section describes the two molecular models discussed here and several observable functions on which the proposed sensitivity
analysis method is validated. A prototypical Lennard-Jones fluid model with two force field parameters and a methane
model with ten parameters are presented. Observables such as the radial distribution function, the mean
square displacement and the pressure spanning from a wide range of model properties are also provided.
All simulations are performed under constant number of atoms, volume and temperature (NVT ensemble).


\subsection{LJ fluid model}
In order to investigate the sensitivity analysis for a realistic system we examine the LJ fluid model.
In this model, the atoms are identical, interacting with the Lennard Jones potential with reduced
non-dimensional parameters $\epsilon_{LJ}=1, \sigma_{LJ}=1$. One of the advantages of the LJ
fluid is that there exists a phase diagram of the reduced density $\rho^{*}$ versus the reduced
temperature $T^{*}$. \cite{Smit} The popularity of this model relies on the generality of systems
of molecular liquids that can be described 
as well as computational efficiency.
 We restrict the force field interactions in the vicinity of cutoff radius $r_{cut}$.
Thus, the (truncated) LJ pair potential is given by
\begin{equation}
  V_{LJ}(r_{ij})=\begin{cases}
    4\epsilon_{LJ} \Big[ \Big(\frac{\sigma_{LJ}}{r_{ij}}\Big)^{12} - \Big(\frac{\sigma_{LJ}}{r_{ij}}\Big)^{6}  \Big] \ , & \text{if $r_{ij}<r_{cut}$}\\
    0 \ ,& \text{otherwise},
  \end{cases}
  \label{LJ_potential}
\end{equation}
while the total potential energy of the system is
\begin{equation*}
	V_{LJ}({q})  = \sum_{\overset{1\leq i,j \leq N}{i < j}} V_{LJ}(r_{ij}) \ ,
\end{equation*}
with
\begin{equation*}
	r_{ij} = |q_i - q_j| = \sqrt{ (q_{i}^{x}-q_{j}^{x})^2 + (q_{i}^{y}-q_{j}^{y})^2 +(q_{i}^{z}-q_{j}^{z})^2} \ ,
\end{equation*}
being the Euclidean distance between the atoms.

Sensitivity analysis is performed on the LJ potential parameters $\epsilon_{LJ}$ and $\sigma_{LJ}$ and
as we show later (see section~\ref{results_section}), the most sensitive parameter is the latter.
We consider a system of $N=2048$ atoms in a cubic simulation box of side length $L = 14.3\sigma_{LJ}$
with periodic boundary conditions (PBC). The reduced temperature of the run is $T^*=0.85 \tau$ which means
that the system is in liquid phase (number density $\rho^*=0.7$). For the numerical scheme, the time-step
is $\Delta t=10^{-3}$ while the length of the run is $10^5$ time-steps. An equilibration period of $10^4$ steps
is sufficient for the fcc lattice to melt and standard reduced units are used throughout the simulations.

\subsection{$CH_4$ model}
\label{methane:subsect}
Methane is a more complicated molecule combined of two different types of atoms; carbon (C)
and hydrogen (H). Active research is targeted on $CH_4$ because of its environmental impact and
energy utilization. \cite{methane_environment} Our sensitivity study is expanded and validated on this more complex
molecular model which consists of different intermolecular potentials between the pairs of atoms (bonded
and non-bonded) as well as additional parameters imposed by the geometry of the molecule (bonds
and angles).  We define $V({q})$ the total potential and $N$ the total number of atoms (both C's and H's).
\begin{equation}
	V({q}) = V_{bond}({q}) + V_{angle}({q}) + V_{LJ}({q}) \ .
\end{equation}

where $V_{bond}({q}), V_{angle}({q})$ are quadratic intramolecular potential functions of the bonds and angles respectively. $V_{LJ}({q})$ is the non-bonded potential as defined in the previous subsection.
For more details concerning the model see Appendix C.
\begin{figure}
	\begin{center}
	\includegraphics[height=10em]{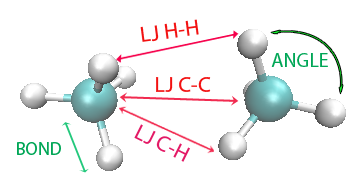}
	\caption{Visualization of the $CH_{4}$ interactions. The site to site non-bonded LJ interactions (intermolecular)
	              are marked in red whereas the intramolecular potential interactions are marked in green.}
	\label{fig:2meth}
 	\end{center}
\end{figure}

The parameter vector, $\theta$, consists of six LJ parameters
(three different LJ potentials depending on the atom type, see also Figure~\ref{fig:2meth}), two bond parameters and
two angle parameters. 
The parameters values of $CH_4$ are summarized in Table \ref{tab:meth_params} whereas the values of the
simulation parameters are presented in Table \ref{tab:sim_params}.
\begin{table}[!hbt]
\centering
\begin{tabular}{ | l || c | r | r | }
  \hline
   & $\epsilon_{LJ} \tiny{[\frac{Kcal}{mol} ]}$ & $\sigma_{LJ}$ \tiny{[\AA]} & $r_{cut}$ \tiny{[\AA]}\\
  \hline
 $C-C$ & 0.0951 & 3.473 & 15.0  \\
 $C-H$ & 0.0380 & 3.159 & 15.0  \\
 $H-H$ & 0.0152 & 2.846 & 15.0\\
 \hline
\end{tabular}
\begin{tabular}{ |c|c|c|c|}
\hline
 $K_{b}$ \tiny{[$\frac{Kcal}{mol \AA^{2}}$]} & $r_0$ \tiny{[$\AA$]}  & $K_{\theta}$ \tiny{[$\frac{Kcal}{mol \cdot deg^{2}}$]} & $\theta_0$ \tiny{[rad]} \\
 \hline
  700 & 1.1 & 100  & 1.909  \\
  \hline
\end{tabular}
\caption{Non-bonded $LJ$ coefficients as well as bond and angle coefficients for methane. \cite{Dreiding}}
\label{tab:meth_params}
\end{table}

\begin{table}[!hbt]
\centering
\begin{tabular}{ | c | c | c | c | c | c | c|}
  \hline
  N(molecules) & T [K] & L [\AA] & $\rho$ [$\frac{moles}{\AA^3}$]  & $\gamma$  \\
  \hline
 512 & 100 & 32.9 & 0.0143  & 0.5  \\
  \hline
\end{tabular}
\caption{Simulation parameters for $CH_{4}$ }
\label{tab:sim_params}
\end{table}

\subsection{Observables}
\label{sec:observ}
To validate the proposed pathwise SA approach we have calculated various observables that are related to thermodynamical, structural and dynamical properties of the molecular stochastic models. These quantities are experimentally tractable and are related to the microscopic as well as the macroscopic level.
In more detail, here, we focus on radial distribution function (RDF), mean
square displacement (MSD) and pressure.
Other studies \cite{Iordanov} in the literature computed observables such as the Helmholtz
free energy, density, enthalpy to name some. Despite the fact that the RDF as well as the pressure are equilibrium
quantities, MSD is related to the dynamics (time-series averaging) making the proposed pathwise method suitable
for such long-time quantities.
Note also, that there are no closed analytic expressions for all the above observables with respect to the force field (model) parameters.

\subsubsection{Radial distribution function}
The structure of liquids is characterized by the pair radial distribution function, $g(r)$, ($g^{(2)}(r)$ to
be more precise) and it is the most important observable of molecular simulations due to the fact that the
ensemble average of any pair function may be expressed by it. \cite{CompSim,McQuarrie} Furthermore, $g(r)$
can be calculated experimentally by X-ray diffraction. \cite{McQuarrie}
The RDF is the pair distribution function that indicates the normalized distribution of a pair of identical atoms
(or molecules) at a given distance. For long intermolecular distance $r$ in liquids, $g(r)$ fluctuates around unity. This static observable is based on the equilibrium structure of the system
and it is constructed by histogram averages. For $N$ identical atoms let the two-atom distribution  function be
\begin{equation}
	P_{N}^{(2)}(q_{1}, q_{2}) := \frac{1}{(N-2)!} \int e^{-\beta V(\bf{q})} dq_{3}\dots dq_{N} \ ,
\end{equation}
where $q_{1}, q_{2}$ are the positions of the first and second atoms kept fixed, irrespective of the configuration of the rest of the particles. For a (homogeneous) liquid, it holds that
\begin{equation}
	P_{N}^{(2)} = \rho^{2} g^{(2)}(|r_{1,2}|),\quad \rho = \frac{N}{\text{Vol}} \
\end{equation}
where $\rho$ is the number density while $Vol$ is the volume of the simulation box. If the atoms were independent of each other, $P^{(2)}$ would equal $\rho^{2}$ so in practice $g(r)$ corrects for the spatial (density) correlation between atoms.
For the $CH_{4}$ model, we consider the molecular $g(r)$ which is based on the center of mass of each individual molecules.

\subsubsection{Mean square displacement}
The mean square displacement associates the diffusion coefficient, $D$, with the atom (or center of
mass for molecules) coordinates and is a measure of the spatial extent of random motion of the Langevin
dynamics. It is defined as
\begin{equation}
	MSD = \langle(q_t-q_{t_0})^2 \rangle = \mathbb E_{Q_{[t_0,t]}}\big[(q_t-q_{t_0})^2\big]
\end{equation}
where $q_{t}, q_{t_0}$ are vectors of particle positions at time $t$ and reference time instant $t_0$, while the brackets,
$\langle\cdot,\cdot\rangle$, denote ensemble averaging over all configurations of all the atoms (or
molecules). This quantity provides us with information about the dynamical properties of the system.
The MSD and the diffusion coefficient, $D$, are related by Einstein's equation
\begin{equation}
	2D = \frac{1}{d} \lim_{t \to \infty} \frac{\partial \langle (q_t-q_{t_0})^2 \rangle}{\partial t }
\label{eq:diff_coeff}	
\end{equation}
Where $d$ is the dimension of the system (here $d=3$).

\subsubsection{Pressure}
Temperature and pressure are macroscopic thermodynamic parameters defined in an experimental
setup but they can also be defined microscopically. Pressure is given by the expression \cite{CompSim}
\begin{equation*}
	P = \frac{\rho}{\beta} + \frac{vir}{\text{Vol}} \ ,
\end{equation*}
where the first term is the kinetic energy contribution while $vir$ is the atomic (or molecular) virial given by
\begin{equation*}
	vir  = \frac{1}{3} \sum_{1\le i \le N}\sum_{j>i} F_{ij}r_{ij} \ .
\end{equation*}
Note that $F_{ij}$ is the total force (both non-bonded and bonded in the $CH_4$ case) between atoms (or molecules)
$i$ and $j$.

\section{Results}
\label{results_section}
Every model at hand has a domain of applicability; i.e., the forcefield representation
allows to calculate (usually thermodynamic) properties of interest in accordance to experimental values
within a margin of error. This means that a force field might represent
well one property, such as density, but may not be valid for others, or might represent all of them less
accurately.
In the following we perform simulations where the RER and FIM for each perturbed variable are computed. Discussion on the results as well as validation with respect to the observable quantities defined in section \ref{sec:observ} supports our results.

\subsection{LJ fluid}

RER and FIM calculations for the LJ fluid are summarized in Figure \ref{fig:RER_new_formula}. We compare the RER value using the continuous time statistical estimators, Eqs. \ref{eq:H_cont}, \ref{eq:F_cont}. The middle bar corresponds to the FIM-based RER whereas the left and right bars are the values of estimator \ref{eq:H_cont} for a negative and positive perturbation by $\epsilon_{0}=5\%$ respectively. All the plots are normalized upon division with the number of particles.
As the figure suggests $\sigma_{LJ}$ is the most sensitive parameter.
Systems size effects have been thoroughly examined by performing test simulations of bigger systems under the same parameters, which produce similar results to those presented here. It has been shown for a similar model that uncertainty in thermodynamic and
transport properties based on the potential parameters is larger than statistical simulation uncertainty. \cite{Pernot}\\

The corresponding results for the discrete time case using the BBK integrator are shown in the Appendix. There's minor discrepancy of order $O(\Delta t)$ as previously mentioned in section ~\ref{pathwise_SA_meth} due to the discretization error bias. We note here that the continuous time computations are faster since the RER and FIM formulas are less complex.

\begin{figure} [h]
\centerline{\includegraphics[height=20em]{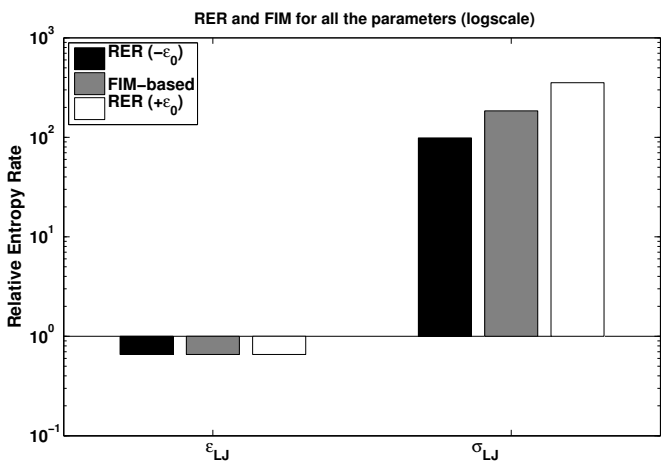}}
\caption{RER and FIM of continuous time estimators (\ref{eq:H_cont}), (\ref{eq:F_cont}). Comparing with the discrete time case (supplementary material), the values are almost identical. $\sigma_{LJ}$ is the most sensitive parameter.}
\label{fig:RER_new_formula}
\end{figure}

\begin{figure} [h]
	\centerline{\includegraphics[height=19em]{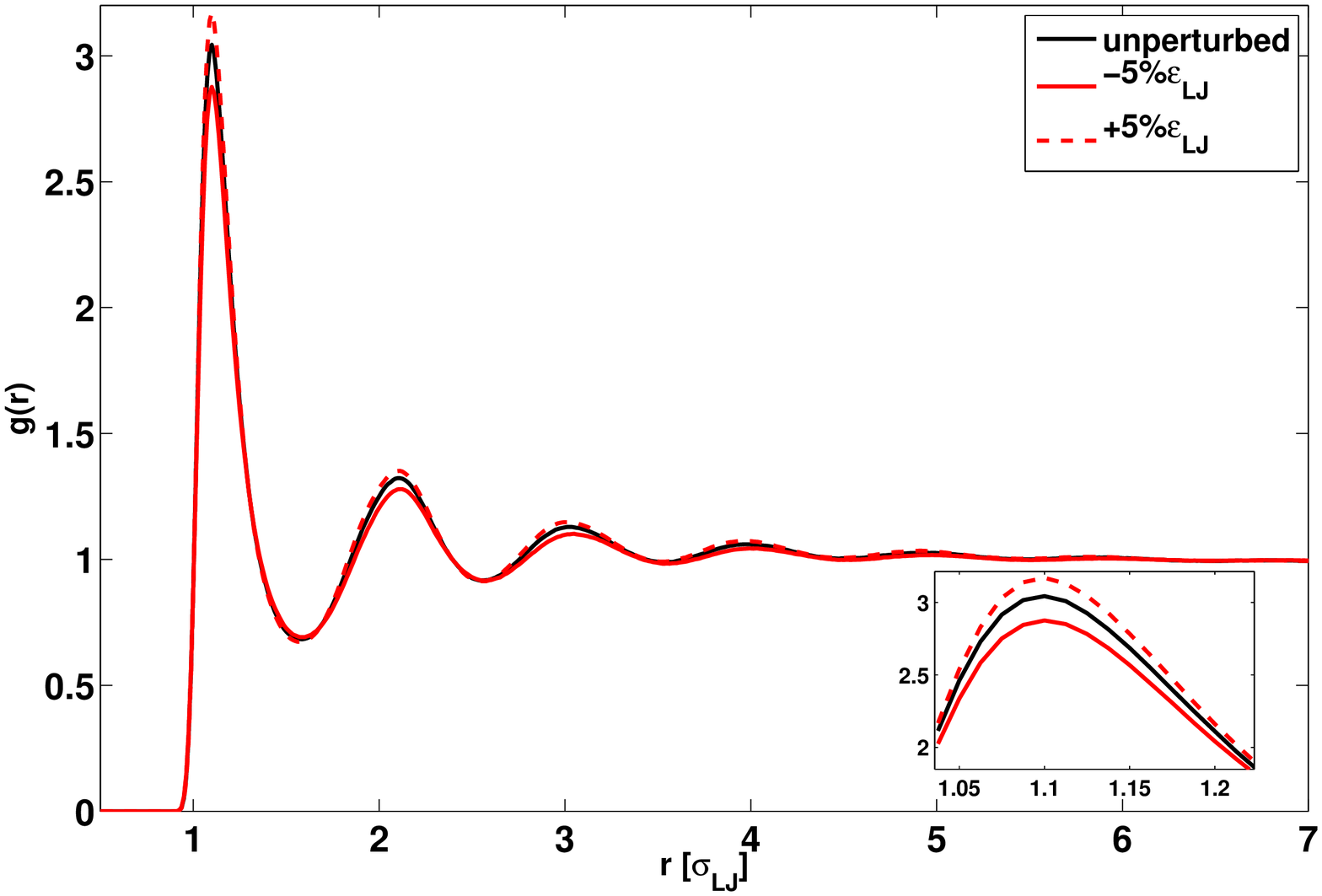}}
	\centerline{\includegraphics[height=19em]{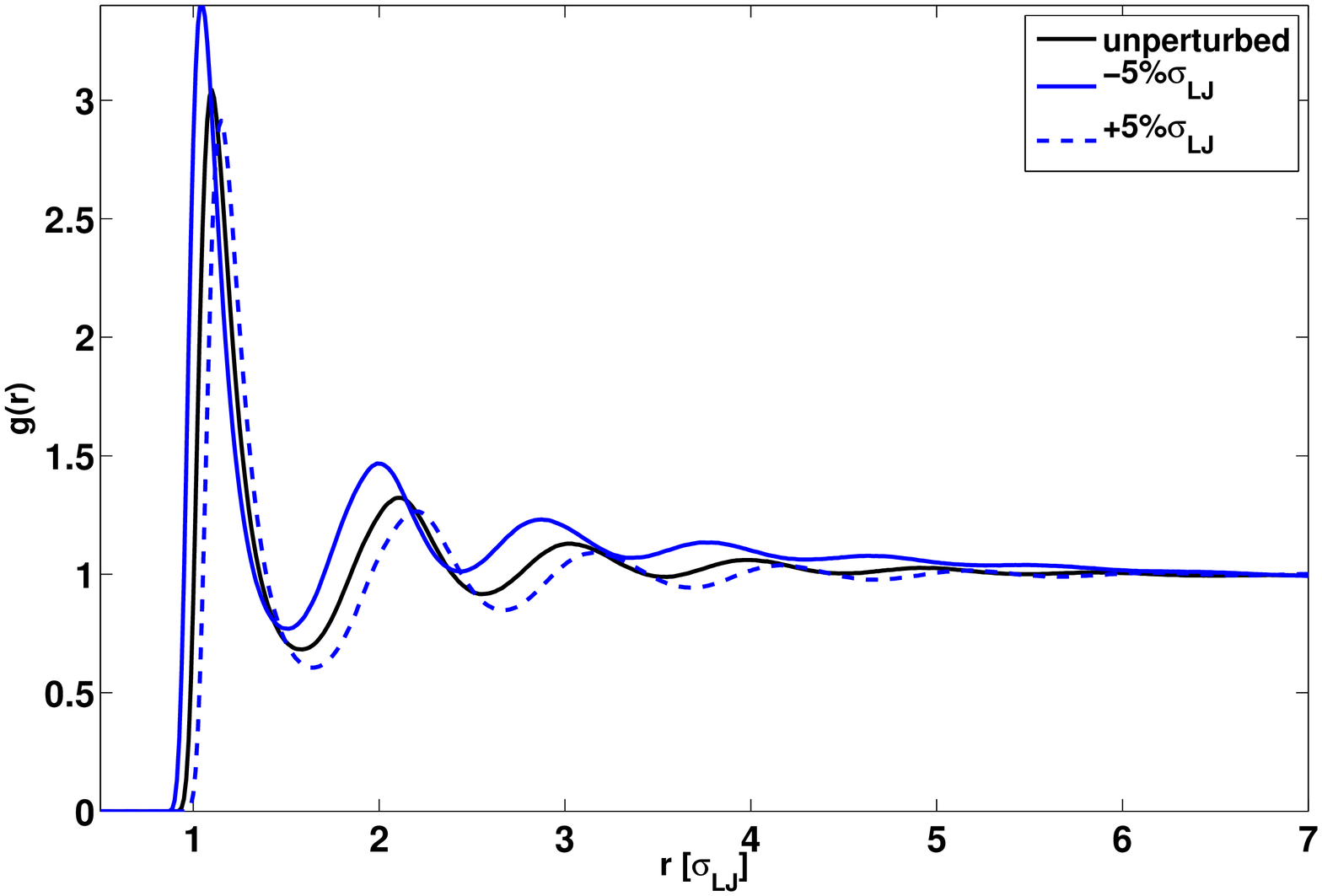}}	
	\caption{Effect of perturbation of $\epsilon_{LJ}$ parameter by $\pm5\%$ (upper panel) and $\sigma_{LJ}$ parameter by $\pm5\%$ (lower panel) on RDF. The first peak is shifted vertically for fluctuations around $\epsilon_{LJ}$ whereas it is shifted to the right or left when the fluctuations concern the $\sigma_{LJ}$ parameter. It is clear that $\sigma_{LJ}$ (lower pannel) is more sensitive as the plots differ substantially, which is in agreement with the RE method.}
	\label{fig:rdf_sig_eps_2}
\end{figure}

\begin{table}[hb]
\centering
\begin{tabular}{ | l || c | c| c| }
  \hline
   perturbation & $ || g^{\theta}(r) - g^{\theta+\epsilon_{0}}(r) ||_{L_{2}}$ & $\frac{|g^{\theta}| - |g^{\theta+\epsilon_{0}}|}{|g^{\theta}|} $& RER \\
  \hline
 $+5\%\epsilon_{LJ}$ & 0.049 &  0.8 \% & 0.79 \\
$-5\%\epsilon_{LJ}$ &   0.066 & -1.17\% & 0.79 \\
 $+5\%\sigma_{LJ}$ & 0.47 &-3.83 \% & 409 \\
 $-5\%\sigma_{LJ}$  &   0.59 & 7.4 \% & 115 \\
 \hline
 $ r_{cut}=1.6\sigma_{LJ}$ & 0.189 & -3.44 \% & 0.71 \\
  $ r_{cut}=7\sigma_{LJ}$   & 0.01 & 0.19\% & $1.6\times 10^{-4}$ \\
  \hline
\end{tabular}
\caption{$L_{2}$ norm of the difference of the unperturbed minus the perturbed g(r) and normalized area difference.}
\label{tab:L2_rdf_LJ}
\end{table}

Validation on the stronger sensitivity on $\sigma_{LJ}$ compared to $\epsilon_{LJ}$, is demonstrated by the RDF ($g(r)$) plots shown in Figure \ref{fig:rdf_sig_eps_2}.
Note that the gradient of the potential, i.e. the interatomic force, depends linearly with respect to $\epsilon_{LJ}$. An increase in this parameter leads to a deeper potential well 
and stronger attraction between the atoms at the same distance. Thus, as expected, positive perturbation in $\epsilon_{LJ}$ leads to an increase of the first peak in the RDF graph. In addition, only the first peak of the $g(r)$ is affected, the rest of the curve remains the same.
Positive (negative) perturbations on the $\sigma_{LJ}$ parameter shift the whole RDF graph due to the fact that the atoms sense greater (weaker) repulsion forces. Hence, the distribution maximum is transferred to a longer (shorter) distance. We also notice that the peak of the curve has increased at the new maximum which can be explained by the finite volume of the same simulation box (NVT ensemble) of the unperturbed system.

In order to get a more detailed insight on the RDFs shown in Figure \ref{fig:rdf_sig_eps_2} we have also computed the $L_{2}$ norm, shown in Table \ref{tab:L2_rdf_LJ}. The $L_{2}$ norm is suitable for a comparison of the unperturbed versus the perturbed plots $g^{\theta}(r)$ and $g^{\theta+\epsilon_{0}}(r)$ respectively. As the RER/FIM computations have shown that there is a relative entropy difference of about 2 orders of magnitude with respect to the two potential parameters (Figure \ref{fig:RER_new_formula}), we now observe a consistent difference, of about 5 times in $L_{2}$, for the RDF observable. Table \ref{tab:L2_rdf_LJ} and Figure \ref{fig:rdf_sig_eps_2} suggest that the positively and negatively perturbed RDF plots exhibit a more symmetric behavior on the $\epsilon_{LJ}$ parameter than the $\sigma_{LJ}$. Moreover $-5\%\sigma_{LJ}$ changes the packing of the LJ fluid completely; all the density distribution peaks are moved to a shorter distance. This result is consistent with Pinsker's inequality (\ref{Pinsker}) as the more sensitive direction allows for greater differences in the expected values of the observables.

Opposite perturbation directions yield different RER values whereas this is not the case for the FIM based RER which is a second-order (quadratic) approximation. In one of the realizations in our example, RER for $+(-)5\%\sigma_{LJ}$ is 360.7 (101.1) and FIM is 196.6 meaning that $\mathcal{H}(Q^{\theta}|Q^{\theta \pm \epsilon_{0}})$ is not symmetric w.r.t. FIM and the negative direction being more sensitive.

There is no analytic formula that relates the MSD to the potential parameters but we expect that a larger deviation will result upon perturbation of a more sensitive parameter. As we can see in Figure \ref{fig:LJ_fl_msd_5perc}, the line for the insensitive perturbed parameter $\epsilon_{LJ}$ slightly differs from the black one, both for positive and negative $\epsilon_{0}$. On the contrary, the line that corresponds to the increased $\sigma_{LJ}$ is further away and under the unperturbed one. Based on the aforementioned discussion on $g(r)$ this is reasonable, as an increase in the $\sigma_{LJ}$ values leads to stronger repulsive forces at the same distance, hence more atom collisions and consequantely to a larger friction coefficient, i.e. lower mobility of the LJ atoms. A decrease in $\sigma_{LJ}$ lowers the interatomic repulsive forces and there's no significant effect at this density because the random forcing dominates the dynamics. This result is consistent with Pinsker's inequality as it provides an upper bound only, meaning that although this parameter is indicated as more sensitive (bigger RER value on the r.h.s.) the expected value w.r.t. this observable slightly changes upon perturbation. Additional runs (realizations of the Markov chain starting from different configurations) for the same negative $\sigma_{LJ}$ direction have shown that the errorbars are within $2.5\%$. The linear dependence of the interatomic forces with respect to $\epsilon_{LJ}$ accounts for increased (decreased) interatomic interaction strength when this parameter is changed upwards (downwards).

\begin{figure}
	\centerline{\includegraphics[height=19em]{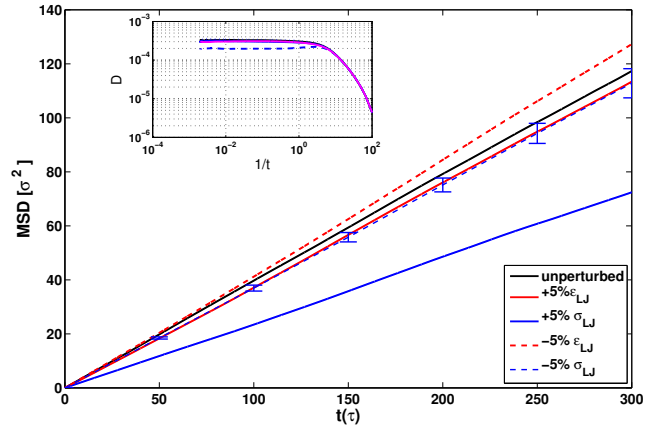}}
	\caption{MSD for different perturbed directions by $\pm5\%$. The $\epsilon_{LJ}$ parameter has a small impact in comparison with the more sensitive $\sigma_{LJ}$. The MSD plot for the positive perturbation of $\sigma_{LJ}$ stands out as the increased collisions dominate the random forcing. Errorbars indicate the standard deviation for $-5\%\sigma_{LJ}$ and the deviation propagates with time. The inset illustrates the diffusion coefficient difference in logscale.}
	\label{fig:LJ_fl_msd_5perc}
\end{figure}

%
%

Table \ref{tab:pressure_and_D_LJ} contains the diffusion coefficient $D$ (from eq. ~\ref{eq:diff_coeff}) related to Figure \ref{fig:LJ_fl_msd_5perc} and depicts a quantitative aspect. The perturbation direction of $+5\%\sigma_{LJ}$ is dominant and clearly results in slowing down the diffusion of the LJ fluid particles.

As mentioned above for simulation of the LJ fluid standard non-dimensional (reduced units) are used. The reduced pressure is denoted by $P^{*}$ and Table \ref{tab:pressure_and_D_LJ} contains the simulation results. Once more we observe a greater influence in perturbation of the parameter $\sigma_{LJ}$ especially for a positive increase. This is consistent with the fact that the volume remains unchanged and the repulsive forces increase as discussed earlier, giving a pressure rise of fourteen times more for a $+5\%$ perturbation. We observe the opposite fact for a reduction in $\sigma_{LJ}$. The $\epsilon_{LJ}$ parameter has a more symmetric influence and this is explained by the linear increase in the forces (derivative of the potential formula) between the atoms and consequently via the virial coefficient it is depicted at the pressure. The third column of Table \ref{tab:pressure_and_D_LJ} compares the relative pressure change with respect to the unperturbed run and the pressure standard deviation is on the column that follows.  Parameter $\sigma_{LJ}$ alters the pressure by an order of magnitude, a result which is consistent with the RER/FIM calculations in the previous subsection.

\begin{table}[hb]
\centering
\begin{tabular}{ | l || c | c |c ||| c| }
  \hline
   perturbation & $ P^{*}$ & $\frac{P^{*}_{\theta+\epsilon_{0}}-P^{*}_{\theta}}{|P^{*}_{\theta}|}$ & $\sigma_{STD}$ &  $ D [\frac{\sigma^{2}_{LJ}}{\tau}]$\\
  \hline
  unperturbed  &  0.11  & -   & 0.25 & $3.2 \times 10^{-4}$\\
 $+5\%\epsilon_{LJ}$ & -0.10  & -1.92 & 0.26 & $2.9 \times 10^{-4}$  \\
 $-5\%\epsilon_{LJ}$ & 0.28 & 1.56 & 0.23  & $ 3.4 \times 10^{-4}$ \\
 $+5\%\sigma_{LJ}$ &   $\bf{1.71}$  &  $\bf{14.34}$ & 0.6  & $\bf{1.8 \times 10^{-4}}$ \\
  $-5\%\sigma_{LJ}$  &   -0.47 & -5.24 & 0.12 & $ 3.04 \times 10^{-4}$ \\
  \hline
\end{tabular}
\caption{(left)Pressure change with respect to different perturbation directions of the LJ fluid parameters. $ \sigma_{LJ}$ is the most sensitive direction. (right)Diffusion coefficient of the MSD plots. The errorbars are within $\pm2.5\%$.}
\label{tab:pressure_and_D_LJ}
\end{table}

\subsubsection{Discontinuous model parameter: cutoff radius}
The $r_{cut}$ is a parameter of the model but the potential is not differentiable with respect to it.
Hence we can compute the relative entropy rate but we cannot have an estimate of the Fisher Information matrix because the computation of FIM involves products of partial derivatives (see also Eq. \ref{eq:F_cont}).
Figure \ref{fig:LJ_fl_RER_rcut} summarizes the quantity $\mathcal{R}({r_{cut}}^{ref}| r_{cut})$ per particle, where in this notation we mean that the RER integral differential is w.r.t. the path space measure corresponding to the model's $r_{cut}$ as reference. The potential tends to zero at distance $r_{cut}=2.5\sigma_{LJ}$, that is a typical value also used in the literature, so information lost upon trimming the potential tail is small in comparison to that when $r_{cut}$ is shifted to the left. We expected that the RER should be higher for a negative perturbation of $r_{cut}$ as validated in Figure \ref{fig:LJ_fl_RER_rcut} and the asymmetry (exponential form for negative perturbation) comes from the formula (plot) of the potential; more information regarding the attractive part is lost rapidly for a $-10\%$ reduction step from the reference $r_{cut}=4\sigma_{LJ}$. Indeed the trick of the $r_{cut}$ convention has been used in molecular simulations in order to reduce the computations at the expense of minimal information loss, so our results using this pseudo-metric indicate that our choice of $r_{cut}$ is suitable.
Additional runs for an increase in $r_{cut}$ suggest a trivial gain of information based on the $\mathcal{H}(Q^{r_{cut}^{ref}}| Q^{r_{cut}})$ value as well as the RDF (see next).\\

\begin{figure} [h]
	\centerline{\includegraphics[height=19em]{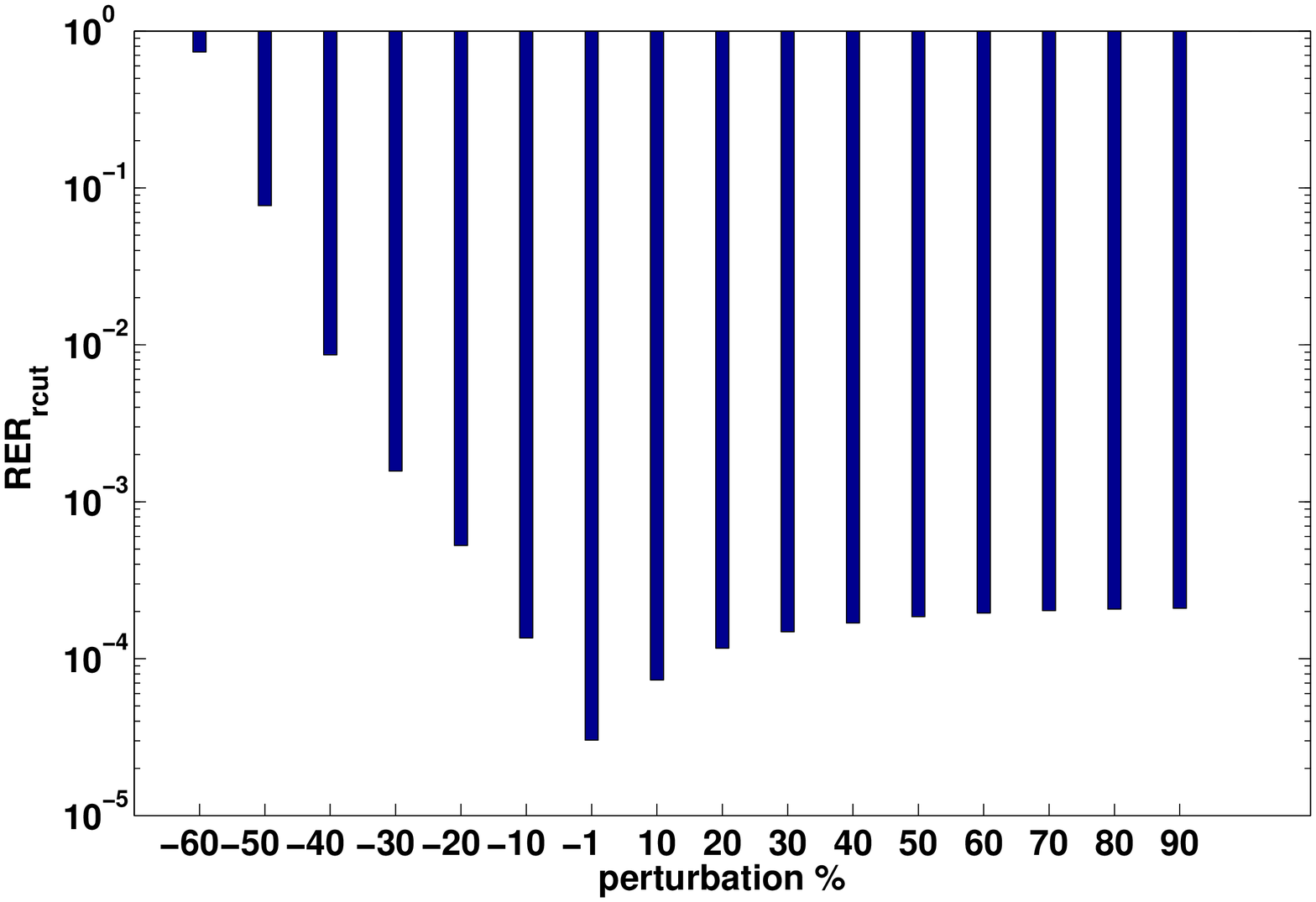}}
	\caption{RER per particle for different $r_{cut}$ values in logscale. Perturbation of -10\% corresponds to the 90\% of the reference $r_{cut} = 4\sigma_{LJ}$. There is significant loss of information when we restrict the potential tail ($r_{cut}$) to less than one half, as that value is near the minimum of the potential well and a fraction of the attractive forces is lost. }
	\label{fig:LJ_fl_RER_rcut}
\end{figure}

The RDF plot changes with a change in $r_{cut}$ as shown in Figure~\ref{fig:LJ_fl_rdf_diff_rcut}. When the potential tail is restricted up to $r_{cut}$, the long-range attractive part is zero after that distance. This results to weaker long-range attractive forces (loss of cohesive energy) hence the first peak in the RDF graph is lower and the mass is distributed to the right. We have included the plot of a $60\%$ decrease to illustrate the higher dependence on a ``premature'' truncation and a plot of $75\%$ increase for comparison. The empirical value of $2.5\sigma_{LJ}$ is adequate for simulations, but a further reduction to 1.6 $\sigma_{LJ}$ results to huge loss of information, especially for the attractive part. On the contrary, if we almost double the reference value of 4$\sigma_{LJ}$ to $7\sigma_{LJ}$, the gain is minimal and this can be seen in Figures \ref{fig:LJ_fl_RER_rcut} and \ref{fig:LJ_fl_rdf_diff_rcut}.

We have seen here that the influence of this parameter is minimal in comparison with the potential parameters in Figure \ref{fig:RER_new_formula} for this reference value in terms of RER. RER per parcticle for a $-5\%\epsilon_{LJ}$ pertrubation is similar to a $-60\%$ redution in $r_{cut}$. The $L_{2}$ norm of the g(r) difference for different $r_{cut}$ values (Table \ref{tab:L2_rdf_LJ} and Figure \ref{fig:LJ_fl_rdf_diff_rcut} ) illustrate the same behavior too. At this point we should stress that the sensitivity of the observables on $r_{cut}$  changes if we choose another reference value;  however in practice usually $r_{cut}$ is not one of the parameters tuned during the force field development/optimization.

\begin{figure} [h]
	\centerline{\includegraphics[height=20em]{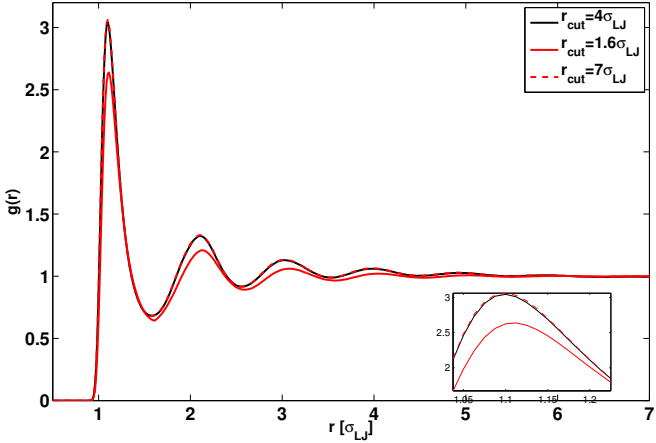}}
	\caption{LJ fluid $g(r)$ for different $r_{cut}$ values. Bigger $r_{cut}$ results to longer range attractive forces binding the atoms closer (higher peak). As the $L_{2}$ norm quantifies, the influence of almost double $r_{cut}$ value ($4\sigma_{LJ}$ is considered as reference), slightly affects the plot and Pinsker inequality validates this fact. On the other hand, a decrease of this parameter leads to loss of information and the corresponding $g(r)$ describes a completely different model. Note that for this plot we increased the system size as the simulation box dimensions restrict the maximum value of $r_{cut}$. }
	\label{fig:LJ_fl_rdf_diff_rcut}
\end{figure}

\subsection*{Non-equilibrium regime LJ fluid}
Finally, we have also studied a non-reversible LJ fluid. In more detail, we have checked the effect of an additional
non-gradient term in the force in the y-direction i.e. $F^\theta(q) = -\nabla V^\theta(q) - G(q), G(q)=[0,\alpha,0,0,\alpha,0,...,0,\alpha,0]^T$. $\alpha=1$ for
the irreversible case and the term $G(q)$ is divergence-free. The eigenvalues and dominant eigenvectors are
summarized in Table \ref{tab:non_rev} and the corresponding RDF plot is given for comparison in Figure~\ref{fig:non_rev}. 
We expected that despite the fact that this process has a different measure close to the stationary measure of the
reversible one, the extra non-gradient term cancels out in eq. (\ref{eq:RER_cont}). Hence our results as expected are
similar but we have demonstrated that the method is general and can be used for a process equipped with a steady
state measure. We aim to the study of more complex systems in non-equilibrium \cite{VH_non_rev_shear2000} in future work.

\begin{table}
\centering
\begin{tabular}{ | l | c || l | c | }
  \hline
   $\alpha = 0$&  & $\alpha = 1$ & \\
   eigenvalues  & eigenvector & eigenvalues &  eigenvector \\
  \hline
  9.434 $\times 10^{4}$  & 0.062 & 1.012  $\times 10^{5}$ & 0.0621\\
  $4.33 \times 10^{11}$  & 0.998 & 4.48 $\times 10^{11}$ & 0.998 \\
  \hline
\end{tabular}
\caption{ eigenvalues and eigenvectors for the non-reversible case}
\label{tab:non_rev}
\end{table}

\begin{figure} [h]
	\begin{center}
	\includegraphics[height=19em]{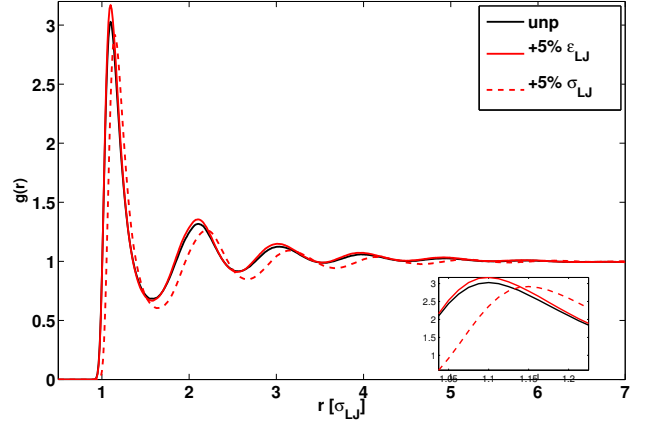}
	\caption{RDF plot for the irreversible case. The norm values are:
	$||g^{\theta +5\%\epsilon_{LJ}}(r) -g^{\theta}(r) ||_{L^{2}}$=0.058, $||g^{\theta +5\%\sigma_{LJ}}(r) -g^{\theta}(r) ||_{L^{2}}$=0.473}
	\label{fig:non_rev}
 	\end{center}
\end{figure}


\subsection{$CH_4$}
\label{CH_4_subsection}

In the following we discuss calculations of RER-FIM as well as various observables for the all-atom methane liquid.
FIM and RER calculations are summarized in Figures \ref{fig:CH4_all_params_RER_FIM} and \ref{fig:CH4_RER_logscale}. In more detail, Figure \ref{fig:CH4_all_params_RER_FIM} shows that the RER values vary orders of magnitude for the various parameter perturbations, hence we grouped them in four panels a)-d). In Figure \ref{fig:CH4_RER_logscale} the FIM-based RER data is plotted in logscale for comparison. We note here that due to the uneven number of the different pairs of $C-C$, $C-H$, $H-H$ we have divided with 8 and 16 the quantities corresponding to the second and third type of pairs in order to obtain comparable plots. All RER values are normalized with the number of corresponding interactions. Furthermore, bigger systems consisting of 4000 molecules conclude with identical results.  \\
As in the LJ paradigm, we can see a greater sensitivity on the $\sigma_{LJ}$ parameters instead of the corresponding $\epsilon_{LJ}$ ones. The errorbars indicate that the variance of the estimators were small and that a positive perturbation increased the value of the RER with respect to the FIM-based RER estimate. Clearly the most sensitive parameter is the $C-H$ bond length $r_{0}$ followed by the bending angle $\theta_{0}$. \\ 
The fact that $r_{0}$ and $\theta_{0}$ are more sensitive is not surprising if we consider that the type of all harmonic potentials is very steep. $K_{b}, K_{\theta}$ constants are of the order $\mathcal{O}(10^{2}-10^{-3})$ as 
obtained from more detailed (ab initio) calculations or from fittings of experimental data (see Table \ref{tab:sim_params}). These constants are part of the $\nabla V_{bond}, \nabla V_{angle}$ which is contained in the estimators.
The asymmetry in the $\sigma_{LJ}$ RER values in comparison to the FIM values (panel b in Fig \ref{fig:CH4_all_params_RER_FIM}) is explained by the third order term contribution in the expansion of RER. A rigorous calculation in Appendix B shows that this term includes  the Hessian of the gradient of the potential w.r.t. the parameters and is non-zero for $\sigma_{LJ}$. \\

\subsubsection{Observables}

We perform the same observable computations as with the LJ fluid model in order to validate the predicted sensitivity of the parameters provided by the RER and pathwise FIM methods. Although we have performed simulations for various values of the parameters, we chose $5\%$ as a suitable value for better representation of our results.
Note that in principe parameter sensitivities change as we change phase space point; in higher temperatures or low densities each observable is affected differently and our proposed RE method incorporates this behavior through the force differences (eq. \ref{eq:RER_cont}). Here we have performed simulations in the temperature range from 80 to 180 K and qualitatively similar results were observed. A more detailed study of SA over various temperatures of more complex (macromolecular) systems will be the subject of a future work.

\begin{figure} 
	\centerline{\includegraphics[height=17em]{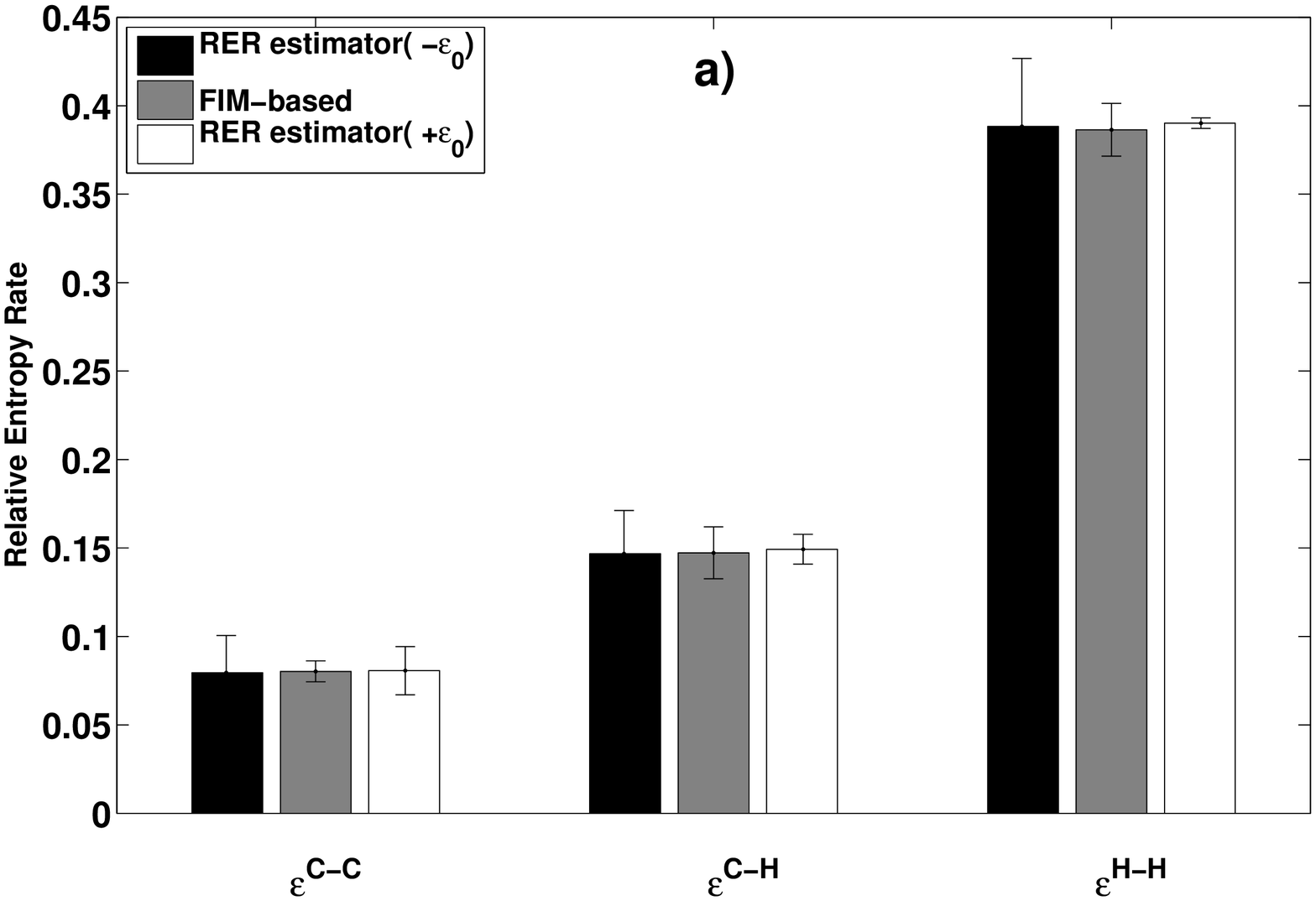}}
	\centerline{\includegraphics[height=17em]{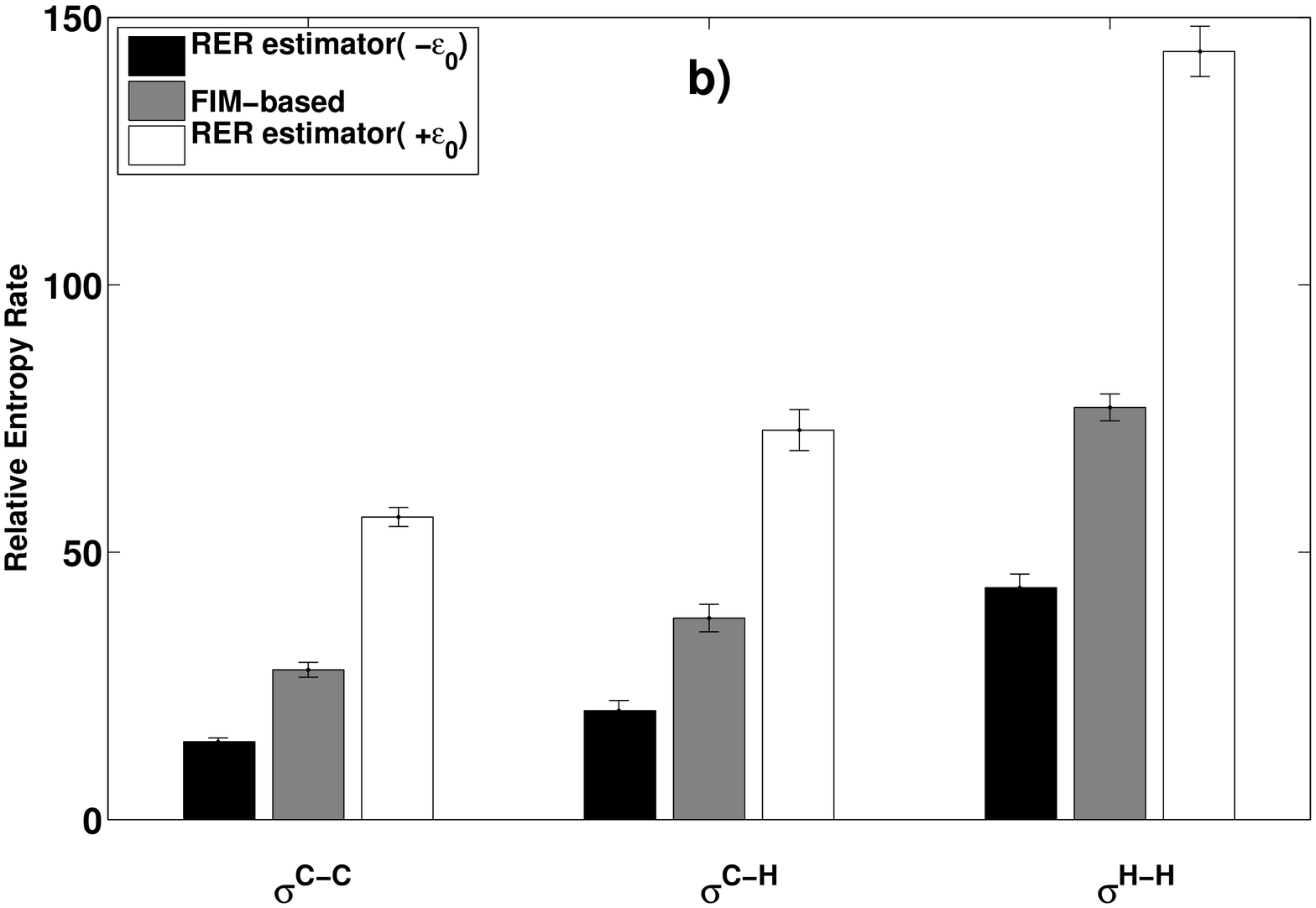}}
	\centerline{\includegraphics[height=17em]{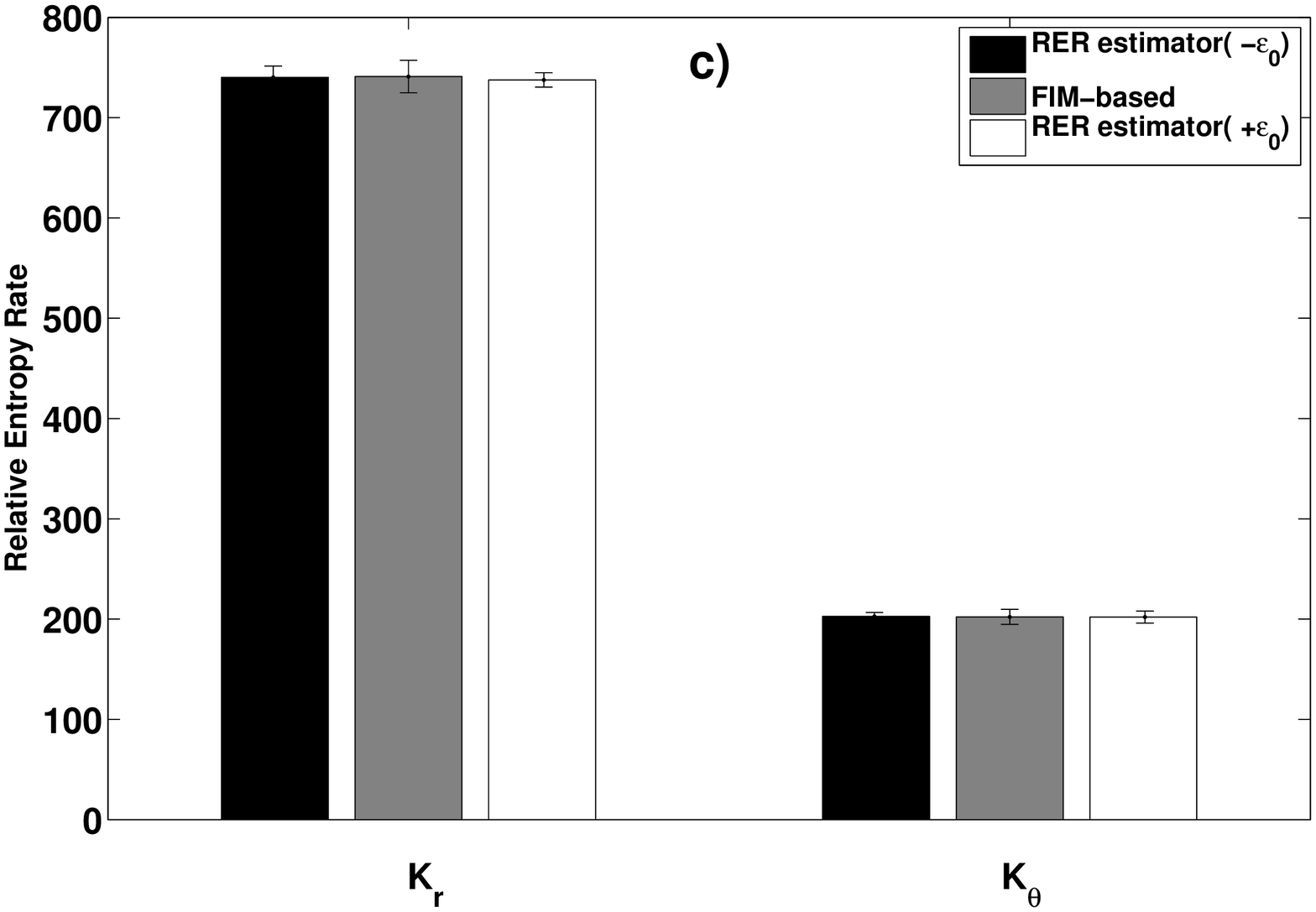}}
	\centerline{\includegraphics[height=17em]{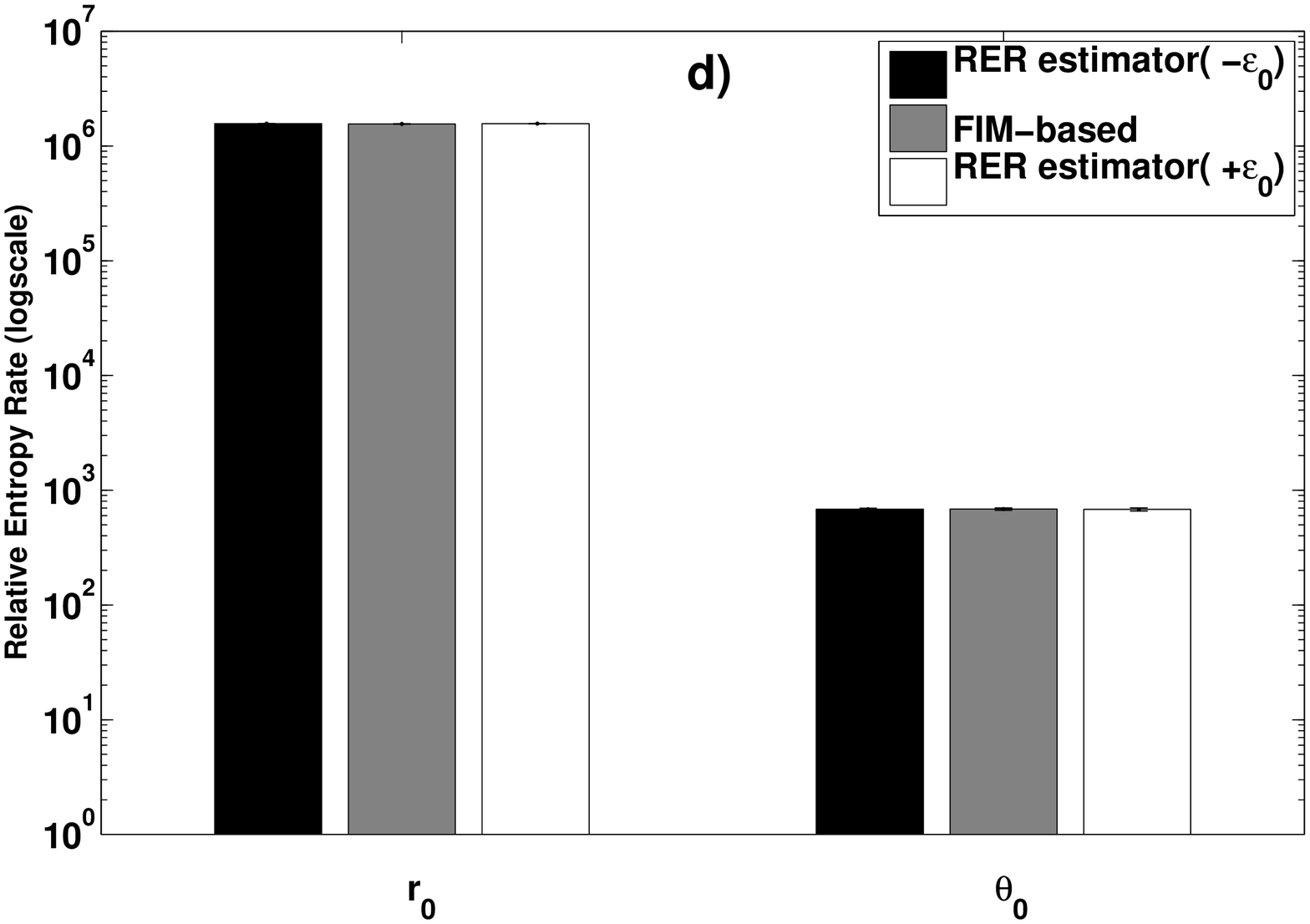}}	
	\caption{$CH_{4}$ per molecule RER-FIM comparison with error bars using the two different estimators for $\pm5\%$ perturbations in all the parameters. Non-bonded (a and b) and bonded (c and d) potential parameters are shown. The parameters are grouped according to their order of magnitude. The most sensitive one is $r_{0}$ followed by $\theta_{0}$ and there has been a minor scaling according to the number of atom-atom pairs. }
	\label{fig:CH4_all_params_RER_FIM}
\end{figure}

\begin{figure} 
	\centerline{\includegraphics[height=20em]{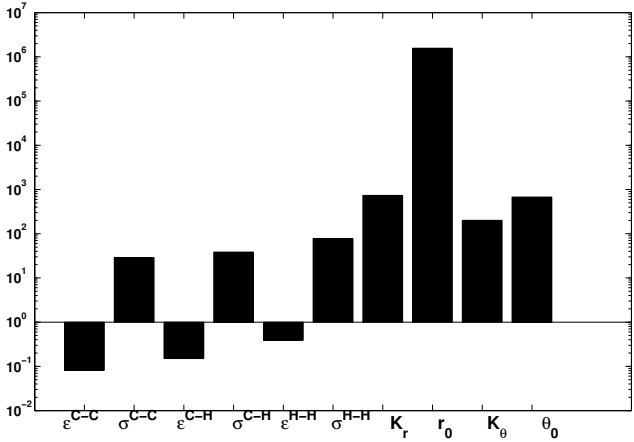}}
	\caption{$CH_{4}$ FIM-based RER comparison for $\pm5\%$ perturbations in logscale. }
	\label{fig:CH4_RER_logscale}
\end{figure}


As in the case of the RDF of the LJ fluid, an increase in the $\sigma_{LJ}$ parameters shifts the graphs to the right (Figure \ref{fig:CH4_rdf_eps}) due to the repulsive forces.
 All the differences with respect to the $L_{2}$ norm are summarized in Table \ref{tab:L2_rdf_CH4} for clarity.

In addition, from the set of RDF data presented in Figure \ref{fig:CH4_rdf_eps}, an increase in $\sigma_{LJ}^{C-H}$ values results to larger deviations. As we keep the volume fixed, the contribution of the $C-H$ interactions in the packing is larger than that of the $C-C$ pairs because of 
the larger number of $C-H$ pairs. Following this graph is the one involving $\sigma_{LJ}^{H-H}$ increase because of the even smaller numerical value in comparison to the other $\sigma_{LJ}$'s. At this point the smaller mass of the hydrogens is the reason although the number of pairs (hence interactions) is the largest.

\begin{figure} 	
	\centerline{\includegraphics[height=20em]{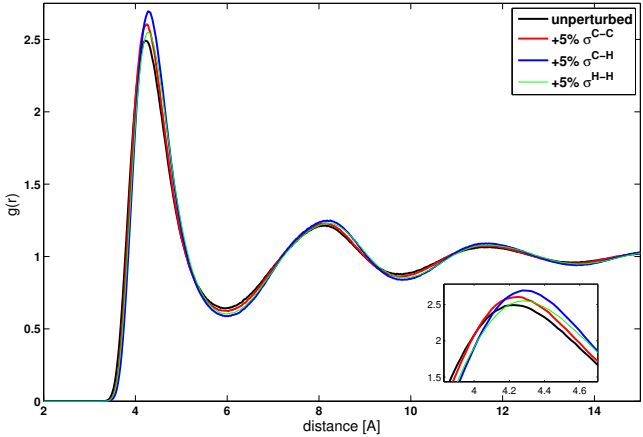}	}
	\caption{$CH_{4}$ molecular g(r) for $+5\%$ perturbations on $\sigma_{LJ}$. The tail of the plot varies slightly hence the zoomed region differs more. As in the LJ fluid case, the $\sigma_{LJ}$ defines the shift of the curve horizontally. }
	\label{fig:CH4_rdf_eps}
\end{figure}

\begin{table} 
\centering
\begin{tabular}{ | l || c | c|}
  \hline
   perturbation & $ || g^{\theta}(r) - g^{\theta+\epsilon_{0}}(r) ||_{L_{2}}$ & $ || g^{\theta}(r) - g^{\theta-\epsilon_{0}}(r) ||_{L_{2}}$ \\
  \hline
 $\epsilon^{C-C}$ & $1.0\times 10^{-2}$ & $1.2\times 10^{-2}$   \\
  $\sigma^{C-C}$ &  $ \bf{1.1\times 10^{-1}}$ & 5.7$\times 10^{-2}$  \\
 $\epsilon^{C-H}$ & $1.7\times10^{-2}$ & $9.6 \times10^{-3}$ \\
 $\sigma^{C-H}$  & $ \bf{2.8\times 10^{-1}} $ & $ \bf{1.7\times 10^{-1} }$ \\
 $\epsilon^{H-H}$  & $ 1.4\times10^{-2} $ & $ 1.15\times10^{-2} $ \\
 $\sigma^{H-H}$  & $\bf{2.05\times 10^{-1} }$ & $\bf{1.6\times 10^{-1} }$  \\
 $K_{b}$  &  $1.1\times10^{-2}$ & $9.7\times10^{-3}$ \\
 $r_{0}$  &  $\bf{1.01\times 10^{-1}}$  &$ \bf{1.1\times 10^{-1}}$\\
 $K_{\theta}$  & $8.7\times10^{-3}$ & $8.2\times10^{-3}$  \\
 $\theta_{0}$  & $9 \times 10^{-3}$ & $9 \times 10^{-3}$ \\
  \hline
\end{tabular}
\caption{$L_{2}$ norm of the difference of the unperturbed minus the perturbed g(r) for $\pm5\%$ perturbation.}
\label{tab:L2_rdf_CH4}
\end{table}

The MSD plots indicate the $\sigma_{LJ}^{CH},\sigma_{LJ}^{HH}$ as the most sensitive parameters. An increase in $\sigma_{LJ}$ results to increased collisions and smaller diffusion coefficient (smaller MSD) as can be seen in Figure \ref{fig:CH4_MSD}. As in the LJ case, positive $\sigma_{LJ}$ perturbations (for all three types) result to greater repulsive forces, hence reduced diffusivity. $\epsilon_{LJ}$ variations slightly affect the MSD with respect to the other parameters and the same holds for $K_{b}$ and $K_{\theta}$ too (we have omitted the plots for brevity). Under this dynamic observable the intramolecular interactions are less relevant than the intermolecular ones, for the specific state point (temperature and density) studied here.

\begin{figure}
	\centerline{\includegraphics[height=19em]{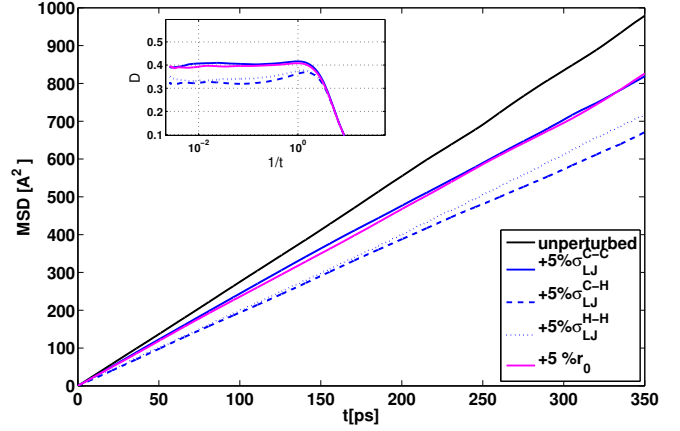}}
	\caption{CH4 MSD for $5\%$ perturbations. We have summarized the most important directions. With respect to this observable, the most sensitive parameter is $\sigma^{C-H}$ followed by $\sigma^{H-H}$. This is in accordance with the RER in Figure ~\ref{fig:CH4_RER_logscale}. The inset illustrates the diffusion coefficient differences in logscale.}
	\label{fig:CH4_MSD}
\end{figure}

Pressure calculations for different perturbation directions are summarized in Table \ref{tab:CH4_pressure}. According to this observable quantity $\sigma_{LJ}^{H-H}$ and $r_{0}$ are the most sensitive parameters, which are also indicated by the RE methods. As in the case of the LJ fluid, a change in $\epsilon_{LJ}$ (in all pair types) affects the pressure less than a change in $\sigma_{LJ}$. Pressure rises through an increase in $\sigma_{LJ}$  due to more atom collisions. Additionally, stronger forces account for a higher pressure virial. The presented results are in accordance with earlier work \cite{SA_on_polarizable_water} on an LJ model of water, 
in which sensitivity analysis using partial derivatives of observables with respect to the parameters were used.
 That study also demonstrated that pressure is greatly affected by variations in $\sigma_{LJ}$ and classified the bond length, $\sigma_{LJ}$ and the bond constant as the most sensitive ones.

A change in the bending angle $\theta_{0}$ does not affect the pressure \cite{pressure_angle} as well as
the impact of the constants $K_{b}, K_{\theta}$ on the pressure is minimal.
We note that the unperturbed system pressure is higher than 1atm because the model we chose (forcefield and integrator) does not reproduce the whole $CH_{4}$ phase diagram precisely. Such small deviations from the equations of state and experiments are expected.
We refer to the supplementary materials for more figures and results which were omitted here for brevity.

\begin{table}[h]
\centering
\begin{tabular}{ | c || c | c | c ||c || c|  }
  \hline
   perturbation & $ P^{\theta+\epsilon_{0}}$[atm] & $\frac{P^{\theta+\epsilon_{0}}-P^{\theta}}{|P^{\theta}|}$ & $\sigma_{STD}$ & $ P^{\theta-\epsilon_{0}}$[atm] & $\frac{P^{\theta-\epsilon_{0}}-P^{\theta}}{|P^{\theta}|}$ \\
  \hline
  unperturbed  &  19.7  & -  & 58.4 &  - & - \\
 $ \epsilon^{C-C}$ & -3.9  & -1.2 & 51.7& 33.1 & +0.7\\
 $\sigma^{C-C}$ &  -2.3&  -1.1 & 53.1 & 87.4 & +3.4\\
 $\epsilon^{C-H}$ &  -31.3   & -2.6 & 52.2 & 63.6 & +2.2\\
 $\sigma^{C-H}$ & $\boldsymbol{177.2}$  & $\boldsymbol{+8}$  & 56.7 & 44.3& +1.2\\
 $\epsilon^{H-H}$ &  8.27   & -0.6  & 49.1 & 23.7& +0.2\\
 $\sigma^{H-H}$ &  $\boldsymbol{437}$ & $\boldsymbol{+21.2}$ & 56.3 & $\boldsymbol{195}$ & $\boldsymbol{+10.9}$\\
 $ K_{b} $ & 15.3  & -0.2  & 49.6 & 14.3 & -0.3 \\
 $ r_{0}  $ & $\boldsymbol{281}$   & $\boldsymbol{+13.3}$ & 56.5 & $\boldsymbol{217}$ &  $\boldsymbol{+12}$\\
 $ K_{\theta} $ & 18.6 & -0.05 & 52.9 & 14.9 & -0.2 \\
 $ \theta_{0} $ & 13.7 & -0.3 & 52.45 & 13.3 & -0.3\\
  \hline
\end{tabular}
\caption{Pressure for $\pm5\%$ perturbation of different directions and the corresponding standard deviation.
The most sensitive parameters $r_{0}$ and $\sigma_{LJ}^{H-H}$  increase the pressure substantially.}
\label{tab:CH4_pressure}
\end{table}



\section{Conclusion}
\label{conclusions_section}
In this work, we extended the parametric SA approach of Ref. \onlinecite{PantazisRER} to stochastic molecular dynamics.
The focus was set particularly to the Langevin equation, however, it is applicable to any molecular system that can be
described by a system of stochastic differential equations. The presented SA approach is based on the relative entropy
per unit time of the path distribution at a reference parameter point with respect to the path distribution at a perturbed
parameter point. Major advantages of this method are that: i) it is capable of handling non-equilibrium steady state systems and ii)
it is computationally tractable through the expansion of the RER which results in the pathwise FIM.
Pathwise FIM provides a fast ``gradient-free'' method for parametric SA since it provides an estimate --up to third-order accuracy--
of the RER for different perturbation directions through a simple matrix multiplication.

We examined two systems; the well-known prototypical LJ fluid and a more complex one: methane ($CH_{4}$).
SA on the LJ fluid system was based on the potential parameters $\epsilon_{LJ}, \sigma_{LJ}$ with the latter being
the more sensitive to perturbations whereas $CH_{4}$ involved 6 intermolecular and 4 intramolecular potential
parameters with the intramolecular parameters being the most sensitive in terms of RER. Static and dynamic
observable quantities such as the radial distribution function, the mean square displacement and the pressure
validated the proposed SA approach. Theoretical justification of the proposed SA approach is also provided
through the Pinsker inequality. We also investigated the effect of the potential cutoff radius, $r_{cut}$, by numerically
computing the RER showing first that RER can be used as an information criterion for assigning appropriate values
to parameters of the system and second that $5\%$ perturbation of $\sigma_{LJ}$ produce greater impact than 
changing $r_{cut}$ from $4\sigma_{LJ}$ to $1.6\sigma_{LJ}$.

Finally, RE for high-dimensional systems was used as a measure of loss of information in coarse-graining.
\cite{Shell_2012,Kats_PP_DT_CGanderror,Emery:98}
Coarse-graining (CG) methods of stochastic systems allow for constructing optimal parametrized Markovian
coarse-grained dynamics within a parametric family, by minimizing the information loss (i.e., the relative entropy) on the
path space. Application of RE to the error analysis of coarse-graining of stochastic particle systems have been
pioneered in these papers. \cite{MK_Majda_Vlachos_2003_PNAS,MK_Majda_Vlachos_2003_JCP,MK_LRB_PP_DT_2008}
Recent ongoing work on application of the RE framework for CG in the non-equilibrium regime where there's
no Gibbs structure can be found in Ref. ~\onlinecite{MK_PP_proof2013}.
We aim to utilize the current SA method to tackle with more complex hybrid macromolecular materials or biomolecular
systems in and out-of equilibrium conditions. \cite{Rissanou13,Harmandaris14,Rissanou14}
Another goal is to adapt the RE method to quantify and indicate the most efficient CG mapping of mesoscale simulations. \cite{Kalligiannaki15}

\begin{acknowledgments}
This research has been co-financed by the European Union (European Social Fund -- ESF) and Greek national
funds through the Operational Program ``Education and Lifelong Learning'' of the National Strategic Reference
Framework (NSRF) -- Research Funding Programs: THALES and ARISTEIA II. The research of MAK was supported  in part by the Office of Advanced
Scientific Computing Research, U.S. Department of Energy under Contract No.
 DE-FG02-13ER26161/DE-SC0010723.
\end{acknowledgments}


\appendix

\section{Pathwise SA at the discrete-time level}
\label{apdx:discrete}

In Section~\ref{pathwise_SA_meth}, we perform SA by first deriving RER and the corresponding pathwise FIM
for the continuous-time stochastic Langevin process and then discretizing the process to get numerical estimates for
these quantities. We can reverse the order of SA and first discretize the Langevin process and then derive the
RER and the pathwise FIM. Here, the latter approach is presented using the BBK algorithm as a numerical integrator
of the Langevin process which defines a discrete-time Markov chain. A preliminary example of this
approach can be found in Ref.~\onlinecite{PantazisRER}. In the BBK integrator, the Hamiltonian part of the
Langevin equation (\ref{eq:system}) is integrated with the Verlet propagator whereas the thermostat is an
Ornstein-Uhlenbeck process and the explicit/implicit propagator is used.

The BBK algorithm \cite{FreeEnergy} reads {\small
\begin{equation}
 \left\{
 \begin{array}{l}
	{p}_{i+\frac{1}{2}} = {p}_i - \nabla V({q}_i) \frac{\Delta t}{2} - \gamma M^{-1} {p}_i \frac{\Delta t}{2} + \sigma \Delta W_i  \\
	{q}_{i+1} = {q}_i + \Delta t M^{-1} {p}_{i+\frac{1}{2}} \\
	{p}_{i+1} =   {p}_{i+\frac{1}{2}}  - \nabla V({q}_{i+1}) \frac{\Delta t}{2}  -\gamma M^{-1} {p}_{i+1}\frac{\Delta t}{2} + \sigma \Delta W_{i+\frac{1}{2}}
\end{array}  \right.
 \label{eq:BBK}
\end{equation}}
$\Delta W_i, \Delta W_{i+\frac{1}{2}} $ are iid Gaussian random vectors with zero mean and covariance matrix
$\frac{\Delta t}{2}I_{dN}$ while $\Delta t$ is the time step of the numerical scheme. Notice that other choices of numerical
integrators can be utilized such as the ones proposed by Leimkuhler et al. \cite{Leimkuhler,Leimkuhler:13} which introduce a
relatively weak perturbative effect on the physical dynamics.

We define the state of the discrete-time system at time-step $i$ as  ${z}_i = ({q}_i, {p}_i) \in\mathbb R^{2dN}$. The process
$\{{z}_i\}_{i=0}^M$ for the BBK integrator is a Markov chain with transition probability $P^\theta(z_i,z_{i+1})$ where
$\theta \in \mathbb{R}^{K}$ is the vector of the system's parameters.
Notice that the length of the discrete-time process is related with the time window of the continuous-time process
through $T=M\Delta t$. The path space probability density, $\bar{Q}_{0:M}^{\theta}(\cdot)$, is defined as
\begin{equation}
	\bar{Q}_{0:M}^{\theta}(\{{z}_i\}_{i=0}^M) = \bar{\mu}^{\theta}({z}^{0}) \prod_{i=0}^{M-1} P^{\theta}({z}_{i}, {z}_{i+1}) \ ,
	\label{eq:path_sp_prob}
\end{equation}
where $\bar{\mu}^{\theta}(\cdot)$ denotes the stationary distribution of the discrete-time. As in the continuous-time
case, we perturb the parameter vector, $\theta$, by adding a perturbation vector $\epsilon_{0} \in \mathbb{R}^{K}$.
At the stationary regime, the pathwise relative entropy of $\bar{Q}_{0:M}^{\theta}$ with respect to $\bar{Q}_{0:M}^{\theta+\epsilon_{0}}$
admits also a decomposition into a linear in time term plus a constant. \cite{PantazisRER} Indeed, it holds that
\begin{equation}
	\mathcal{R}(\bar{Q}_{0:M}^{\theta}| \bar{Q}_{0:M}^{\theta+\epsilon_{0}}) = M \mathcal{H}(\bar{Q}^{\theta}| \bar{Q}^{\theta+\epsilon_{0}})
	+ \mathcal{R}(\bar{\mu}^{\theta}| \bar{\mu}^{\theta+\epsilon_{0}}) \ ,
	\label{eq:Rel_ent_RER}
\end{equation}
where $\mathcal{R}(\bar{\mu}^{\theta}| \bar{\mu}^{\theta+\epsilon_{0}})$ is the relative entropy between the stationary distributions
while $\mathcal{H}(\bar{Q}^{\theta}| \bar{Q}^{\theta+\epsilon_{0}})$ is the RER of the discrete-time Markov chain given by {\small
\begin{equation}
	\mathcal{H}(\bar{Q}^{\theta}| \bar{Q}^{\theta+\epsilon_{0}}) = \mathbb{E}_{\bar{\mu}^{\theta}}
	\Bigg[ \int_{\mathbb R^{2dN}} P^{\theta}({z},{z'}) \log \frac{P^{\theta}({z},{z'})}{P^{\theta+\epsilon_{0}}({z},{z'})} d{z'} \Bigg] \ .
	\label{eq:RER_discrete}
\end{equation}}
The discrete-time RER is related with the continuous-time RER through \cite{Kalligiannaki15}
\begin{equation}
\mathcal{H}({Q}^{\theta}| {Q}^{\theta+\epsilon_{0}}) = \lim_{\Delta t\rightarrow0} \frac{1}{\Delta t}\mathcal{H}(\bar{Q}^{\theta}| \bar{Q}^{\theta+\epsilon_{0}}) \ .
\end{equation}

As expected, discrete-time RER is locally a quadratic functional in
a neighborhood of $\theta$ hence its curvature around $\theta$, defined
by the Hessian, is the pathwise FIM which is given by \cite{PantazisRER} {\small
\begin{equation}
\begin{aligned}
	& F_{\mathcal{H}}(\bar{Q}^{\theta}) = \\
	& \mathbb{E}_{\bar{\mu}^{\theta}}  \Bigg[ \int_{\mathbb R^{2dN}} P^{\theta}({z}, {z'})\nabla_{\theta} \log P^{\theta}({z}, {z'}) \nabla_{\theta} \log P^{\theta}({z},{z'})^T d{z'} \Bigg] \ .
	\label{eq:FIM_discrete}
\end{aligned}	
\end{equation}}
We refer to (\onlinecite{PantazisRER}) for statistical estimators of the discrete-time RER and the 
corresponding pathwise FIM while in Supplementary Materials we provide detailed formulas for
the numerical calculation of (\ref{eq:RER_discrete}) and (\ref{eq:FIM_discrete}) for the BBK integrator.

\section{Expansion of the continuous-time RER}
We now expand the RER in eq. (\ref{eq:RER_cont}) through Taylor series expansion around the point $\theta$.
We start with expanding the $m$-th component of the force, $F^{\theta+\epsilon_0}(q)$, around $\theta$
\begin{equation}
	F_m^{\theta+\epsilon_0}({q}) = F_m^{\theta}({q}) + \nabla_{\theta}F_m^{\theta}({q}) \epsilon_0
	+ \frac{1}{2}\epsilon_0^{T}\nabla^{2}_{\theta} F_m^{\theta}({q})\epsilon_0 + \mathcal{O}(|\epsilon_0|^3)
\end{equation}
where $\nabla$ denotes the $1\times K$ gradient vector while $\nabla^{2}$ denotes the $K\times K$ Hessian matrix.
Then, the RER is written as {\small
\begin{equation*}
\begin{aligned}
& \mathcal{H}(Q^{\theta} | Q^{\theta+\epsilon_0}) \\
&= \frac{1}{2} \mathbb{E}_{\mu^\theta}[ ( F^{\theta+\epsilon_0}({q}) - F^\theta({q}))^T (\sigma\sigma^T)^{-1} ( F^{\theta+\epsilon_0}({q}) - F^\theta({q}))] \\
&= \frac{1}{2} \sum_{m,n=1}^{dN} \mathbb{E}_{\mu^\theta} [(F_m^{\theta+\epsilon_0}({q}) - F_m^\theta({q})) ((\sigma\sigma^T)^{-1})_{m,n}
(F_n^{\theta+\epsilon_0}({q}) - F_n^\theta({q}))] \\
&= \frac{1}{2} \sum_{m,n=1}^{dN} ((\sigma\sigma^T)^{-1})_{m,n} \mathbb{E}_{\mu^\theta} [\nabla_{\theta}F_m^{\theta}(q)\epsilon_0
\nabla_{\theta}F_n^{\theta}(q)\epsilon_0] \\
&+ \frac{1}{2} \sum_{m,n=1}^{dN} ((\sigma\sigma^T)^{-1})_{m,n} \mathbb{E}_{\mu^\theta} [\nabla_{\theta}F_m^{\theta}(q)\epsilon_0
\epsilon_0^{T}\nabla^{2}_{\theta} F_m^{\theta}(q)\epsilon_0] + O(|\epsilon_0|^4) \ .
\end{aligned}
\end{equation*}}
The pathwise FIM comes from the second-order term while the third-order term defines a tensor matrix.

For the LJ non-bonded potential, the leading term of the second-order term (i.e., the pathwise FIM) in the RER
expansion when $\sigma_{LJ}$
is perturbed is of order $O\big(\big( \frac{\sigma_{LJ}}{r} \big)^{\bf{10}}\big)$ while the leading term of the third-order
term of RER is of order $O\big(\big( \frac{\sigma_{LJ}}{r} \big)^{\bf{9}}\big)$ with (typically) $\sigma_{LJ}<r$. 
The fact that the leading term of the third-order term has smaller exponent compared to the second-order term,
makes the contribution of the third-order term to the value of RER significant on average. Therefore, the asymmetry between
$\mathcal{H}(Q^{\theta} | Q^{\theta+\epsilon_0})$ and $\mathcal{H}(Q^{\theta} | Q^{\theta-\epsilon_0})$ observed
both in the LJ fluid (Figure~\ref{fig:RER_new_formula}) and the methane (Figure~\ref{fig:CH4_all_params_RER_FIM})
stems exactly from the significance of the third-order term. Notice that asymmetries
between positive and negative perturbations are not rare and have been observed in biological reaction models and one
method that is employed for assessing parameter identifiability in non-linear models is the profile likelihood method.
\cite{Raue_profile_likelihood}

\section{Potential energy terms of $CH_4$}
In this section, the details of the total potential $V({q})= V_{bond}({q}) + V_{angle}({q}) + V_{LJ}({q})$
for the methane model are presented. The total bond potential equals to
\begin{equation}
	V_{bond}({q}) = \sum_{\mathcal{A}} V_{bond}(|q_j - q_i|)
\end{equation}
where
\begin{equation*}
\begin{split}
	\mathcal{A} = \{\text{$q_{i}$=C, $q_{j}$=H $q_{i},q_{j}\in$ same $CH_{4}$},\\
	\hfill \text{4 bonds per $CH_{4}$} \}
\end{split}	
\end{equation*}
while the local bond potential is
\begin{equation}
	V_{bond}(|q_j - q_i|) = V_{bond}(r_{ij}) = \frac{1}{2} K_b (r_{0} - q_{ij} )^2 \ .
\end{equation}
The two constants $r_0$ and $K_b$ determine the distance and the strength of the bond between the
two atoms, respectively.

The angle defined for each triplet $H-C-H$ on the same $CH_{4}$ molecule is denoted by $\theta_{jik}$. Then, the total
angular potential is
\begin{equation}
	V_{angle}({q}) =  \sum_{\mathcal{B}} V_{angle}(\angle{q_j q_i q_k})
\end{equation}
where
\begin{equation*}
\begin{split}	
	\mathcal{B}= \{\text{$q_{i}$=C, $q_{j},q_{k}$=H, $q_{i},q_{j},q_{k} \in$ same $CH_{4}$},\\
	 \hfill \text{6 angles per $CH_{4}$}\} \ ,
\end{split}	
\end{equation*}
while the local angular potential is given by
\begin{equation}
	V_{angle}(\angle{q_j q_i q_k}) = V_{angle}(\theta_{ijk}) = \frac{1}{2} K _{\theta} (\theta_{0} - \theta_{ijk} )^2 \ .
\end{equation}
The two constants $\theta_0$ and $K _{\theta}$ determine the degree and the strength of the angle, respectively.

Moreover, the non-bonded term of the potential energy, $V_{LJ}({q})$, is given by
\begin{equation}
	V_{LJ}({q}) = \sum_{\mathcal{C}} V_{LJ}(|q_j-q_i|)
	\label{eq:VLJ}
\end{equation}
where
\begin{equation*}
	\mathcal{C} = \{\text{$q_{i}, q_{j}$=H or C, $q_{i},q_{j}\in$ different $CH_{4}$} \}
\end{equation*}
while the functional form of the LJ potential, $V_{LJ}(r_{ij})$, is given by (\ref{LJ_potential}).

Since the LJ potential is the non-bonded term, the sum in (\ref{eq:VLJ}) is over all the atoms of the other methanes.
It is convenient furthermore to divide this sum into three sums, each one corresponding on a different class of
interactions between $C-C, C-H, H-H$. Thus, we can rewrite
 \begin{equation}
	V_{LJ}({q}) = \sum_{\mathcal{C}_{1}} V_{LJ}^{C-C}(r_{ij})
		   + \sum_{\mathcal{C}_{2}} V_{LJ}^{H-H}(r_{ij}) + \sum_{\mathcal{C}_{3}} V_{LJ}^{H-C}(r_{ij}) \ ,
\end{equation}
where
\begin{align*}
	\mathcal{C}_{1} &= \{\text{$q_{i}, q_{j}$= C} \}  \\
	\mathcal{C}_{2} &= \{\text{$q_{i}$= C, $q_{j}$= H, $q_{i},q_{j}\in$ different $CH_{4}$} \}  \\
	\mathcal{C}_{3} &= \{\text{$q_{i}, q_{j}$= H, $q_{i},q_{j}\in$ different $CH_{4}$} \} \ .
\end{align*}
Each LJ potential has its own parameter values.

\bibliographystyle{unsrt}
\addcontentsline{toc}{chapter}{References}
\bibliography{refs_from_bibdesk}

\end{document}